\documentclass[onecolumn]{IEEEtran}
\usepackage{amsmath,amsfonts}
\usepackage{algorithmic}
\usepackage{algorithm}
\usepackage{array}
\usepackage{subcaption}

\usepackage{textcomp}
\usepackage{stfloats}
\usepackage{url}
\usepackage{verbatim}
\usepackage{graphicx}
\usepackage{cite}
\usepackage{comment}

\usepackage{setspace}
\usepackage{amsmath}
\usepackage{amssymb}
\usepackage{mathrsfs}
\usepackage{amsthm}
\usepackage{multirow}
\usepackage{color}
\usepackage{array}
\usepackage{url}
\usepackage{comment}
\usepackage{enumerate,enumitem}
\usepackage{eucal}
\usepackage[normalem]{ulem}
\usepackage{graphicx}
\usepackage{tabto}
\usepackage{amsmath}
\usepackage{bbm}
\usepackage{amsthm}
\usepackage{amssymb}
\usepackage{enumitem}
\usepackage{enumerate}
\usepackage{mathtools}
\usepackage{esvect}
\usepackage{caption}
\usepackage{float}
\usepackage{multicol}
\usepackage{tikz}
\usepackage{soul}
\usepackage{balance}
\usepackage{adjustbox}
\usepackage{booktabs}
\usepackage{mathtools,amssymb}
\usepackage{newtxtext}
\usepackage{newtxmath}
\usepackage{longtable}

\usepackage{hyperref}
\newcommand\myshade{70}
\hypersetup{
	linkcolor  = black!\myshade!black,
	citecolor  = blue!\myshade!black,
	urlcolor   = blue!\myshade!black,
	colorlinks = true,
}

\usepackage{cleveref}

\usepackage{tabularx}

\usepackage{framed}

\usepackage{tikz}
\usetikzlibrary{arrows, patterns, shapes.arrows, decorations.pathmorphing}

\IEEEoverridecommandlockouts

\theoremstyle{plain}
\newtheorem{theorem}{Theorem}
\newtheorem{corollary}[theorem]{Corollary}
\newtheorem{lemma}[theorem]{Lemma}

\theoremstyle{definition}
\newtheorem{definition}{Definition}
\newtheorem{example}[definition]{Example}

\newtheorem{remark}[definition]{Remark}

\newcommand{\FF}{\mathbb{F}}
\newcommand{\BB}{\mathbb{B}}

\newcommand{\GF}{{\mathrm{GF}}}

\DeclareMathAlphabet{\mathbfsl}{OT1}{ppl}{b}{it} 

\newcommand{\cS}{\mathcal{S}}
\newcommand{\cA}{\mathcal{A}}

\newcommand{\cC}{\mathcal{C}}

\newcommand{\cF}{\mathcal{F}}

\newcommand{\cT}{\mathcal{T}}
\newcommand{\cU}{\mathcal{U}}
\newcommand{\cW}{\mathcal{W}}
\newcommand{\cV}{\mathcal{V}}

\newcommand{\cY}{\mathcal{Y}}

\newcommand{\cI}{\mathcal{I}}

\usepackage{balance}

\newcommand{\vlambda}{{\pmb{\lambda}}}

\newcommand{\rs}{{\rm RS}}
\newcommand{\grs}{{\rm GRS}}
\newcommand{\tr}{{\rm Tr}}

\newcommand{\vspan}{{\rm span}}

\newcolumntype{Y}{>{\centering\arraybackslash}X} 

\title{Trace Repair Never Loses to Classical Repair:\\ Exact and Explicit Helper Nodes Selection 
}

\author{

   \IEEEauthorblockN{Wilton Kim, Stanislav Kruglik, Han Mao Kiah}
	
}

\begin{document}
\date{}

\maketitle

\begin{abstract}
Repairing Reed-Solomon codes with low bandwidth is a central challenge in distributed storage. Following the trace-repair framework of Guruswami and Wootters (2017), recent works by Lin (2023) and Liu--Wan--Xing (2024) provided significant improvements in bandwidth using two distinct ideas. Lin constructed a trace-repair scheme that requires no contribution from a set of predetermined nodes $\cS$, while Liu--Wan--Xing identified linear dependencies among the downloaded traces, relating the number of dependent traces to the dimension of a subspace $\cW_k$.
In this work, we fully utilize and unify these ideas. We compute the exact dimension of $\cW_{k,\cS}$ (a generalization of $\cW_k$). We identify the trade-off between the set size $|\cS|$ and the dimension $\dim(\cW_{k,\cS})$. We provide an algorithm to find the combination that results in the lowest bandwidth. Furthermore, we provide an explicit choice of the helper nodes for the repair. Finally, we prove that our optimized scheme never loses to the classical repair scheme, establishing a bandwidth guarantee of at most $k\log|\FF|$ bits for all dimension $k$ and field $\FF$, whenever the trace repair is applicable.
\end{abstract}
\begin{IEEEkeywords}
Reed–Solomon codes, distributed storage, trace repair, single erasure repair, repair bandwidth reduction.
\end{IEEEkeywords}
\section{Introduction}

\IEEEPARstart{R}{eed-Solomon} (RS) codes \cite{reedsolomon} are widely used in distributed storage because all information symbols can be recovered by downloading any $k$ available code symbols (see \cite{dinh2022practical} for a survey). More precisely, let $\FF=\GF(p^{mt})$ and let $\cA=\{\alpha_1,\ldots,\alpha_n\}\subseteq\FF$ be $n$ distinct evaluation points. 
Given $k$ information symbols $\boldsymbol{x}\in\FF^k$, we encode them as a polynomial $c$ of degree at most $k-1$ and store $\boldsymbol{c}=(c(\alpha))_{\alpha\in\cA}$. 
The MDS property then states that any $k$ coordinates of $\boldsymbol{c}$ uniquely determine $c$ and hence $\boldsymbol{x}$. Equivalently, at most $n-k$ erasures can be corrected by downloading any $k$ surviving code symbols. In this paper, we call this the \emph{classical} repair scheme. Its bandwidth is $k\log|\FF|$ bits, since each downloaded symbol lies in $\FF$. While the classical repair scheme is optimal to fully reconstruct a codeword with $n-k$ erasures, it is excessive when there is only one erased code symbol.

In \cite{guruswamiwooters2017}, Guruswami and Wootters proposed a repair scheme for {\em a single erased code symbol} $c(\alpha^*)$ by utilizing the trace function $\tr:\FF\to\BB$ for some base field $\BB=\GF(p^m)$. 
Specifically, we download $n-1$ traces of the form $(\tr(\lambda_\alpha c(\alpha)/(\alpha-\alpha^*)))_{\alpha\in\cA\setminus\{\alpha^*\}}$ for some $\lambda_\alpha\in\FF$. 
This then results in a bandwidth of $(n-1)\log|\BB|$ bits. In terms of bandwidth, the Guruswami-Wootters scheme outperforms the classical scheme only when $k> (n-1)/t$. There is a flurry of works utilizing the Guruswami-Wootters scheme in different setups \cite{Duursma2017, Dau2018, Dau2018io, Mardia2019, Tamo2019, Chen2019, Weiqi2019, Dau2021, Berman2022, kiah2024, dinh2024, kim2024decoding, kruglik2025}. Despite the versatility of trace repair in various settings, a critical gap persists in the standard setup. Specifically, whether there exists a repair scheme that improves upon the classical repair scheme for all values of $k\le n-q^{t-1}$ remains open.


Dau and Milenkovic \cite{Dau2017} derived a lower bound for the repair bandwidth. Recent advancements by Lin \cite{Lin2023} and Liu et al. \cite{Liu2024} have identified two distinct ways to reduce bandwidth:
\begin{itemize}[leftmargin=*]
    \item \textbf{Linear Dependence on Traces.} Liu et al. \cite{Liu2024} identified linear dependency on $d$ downloaded traces in the Guruswami-Wootters scheme. Specifically, Liu \cite{Liu2024} related the number $d$ to the dimension of a subspace $\cW_k$ (see Theorem~\ref{thm:liu} for the exact statement) and in the same paper, provided lower bounds on $\dim(\cW_k)$ for some specific finite fields $\FF$.
    \item \textbf{Zero-forcing.} Lin \cite{Lin2023} proposed to consider a polynomial $g(x)$ that vanishes at some predetermined evaluation points in  $\cS$. By multiplying $g(x)$ to the polynomial corresponding to the dual code, it allows us to exclude the code symbols at these evaluation points from the repair scheme entirely.
\end{itemize}
We discuss the previous works further in Section~\ref{sec:prelim}. Although both methods significantly reduce repair bandwidth, they can be further improved through a more comprehensive analysis. Specifically,
\begin{itemize}[leftmargin=*]
    \item Liu et al. \cite{Liu2024} provided only lower bounds on the dimension of $\cW_k$ for some specific finite fields $\FF$, leaving the exact capability of their method unknown. Furthermore, their approach did not incorporate the zero-forcing mechanism introduced by Lin \cite{Lin2023}.
    \item Conversely, Lin \cite{Lin2023} did not exploit the linear dependencies among downloaded traces identified by Liu et al \cite{Liu2024}. Instead, Lin focused entirely on zero-forcing by setting the number of excluded nodes $|\cS|$ to its maximum possible value. However, this choice imposes a strict degree constraints for the parity-check equations. As a result, it is almost impossible to have linear dependency on the downloaded traces.
\end{itemize}

\subsection{Our Contribution}
In this work, we fully utilize the ideas proposed by Lin \cite{Lin2023} and Liu et al. \cite{Liu2024} to further reduce the repair bandwidth. We consider repair problem for a single erased code symbol in a full-length Reed-Solomon code. Without loss of generality, we assume that $c(0)$ is erased. First, we prove the following Theorem~\ref{thm:main_theorem} in Section~\ref{sec:main_idea}.

\begin{theorem}\label{thm:main_theorem}
    Let $\FF = \GF(q^t)$ be an extension field of the base field $\BB = \GF(q)$ for some prime power $q$. Let $n = |\FF|$. Fix $k$ and $\cS$, and consider $(c(\alpha))_{\alpha\in\FF}\in\rs(\FF,k)$ with $c(0)$ being erased. 
Let
    $$\cW_{k,\cS} = \{(h(\alpha))_{\alpha\in\FF}:h(x) = g(x)f(x)/x, f:\FF\to\BB, h(0) = 0, \deg(h(x))\le n-k-1\},$$
    with $g(x) = \prod_{\beta\in\cS} (x-\beta)$ for some nonempty $\cS\subset\FF$ and $g(x) = 1$ if $\cS$ is empty.
    Then, we can repair $c(0)$ with bandwidth of $(n-1-\dim(\cW_{k,\cS})-|\cS|)\log|\BB|$ bits. 
\end{theorem}

Theorem~\ref{thm:main_theorem} guarantees the existence of repair scheme of bandwidth $(n-1-\dim(\cW_{k,\cS})-|\cS|)\log|\BB|$ bits. If $\cS$ is empty, then Theorem~\ref{thm:main_theorem} is reduced to the existence of repair scheme shown in Liu et al. \cite{Liu2024}. Crucially, our work moves beyond existence proof. We summarize our contributions as follows.
\begin{enumerate}[leftmargin=*]
    \item \textbf{Exact Dimension Computation.} Given $k$ and $\cS$, we determine the exact dimension $\dim(\cW_{k,\cS})$ for any finite field $\FF$ by utilizing the relation between \textit{cyclotomic cosets} and polynomials corresponding to $\cW_{k,\cS}$ (Theorem~\ref{thm:dim_of_WkS}).
    \item \textbf{Optimal Trade-off Algorithm.} We generalize the work by Lin \cite{Lin2023} by relaxing the constraint on $|\cS|$. We identify a trade-off between $|\cS|$ and $\dim(\cW_{k,\cS})$. Note that, the number of excluded nodes from the repair scheme is the sum $R_{\cS} = |\cS| + \dim(\cW_{k,\cS})$. Hence, we provide a pruning algorithm (Algorithm~O) that outputs $|\cS|$ and $\dim(\cW_{k,\cS})$ which maximize the sum $R_{\cS}$.
    \item \textbf{Explicit Repair Scheme.} With the output $|\cS|$ and $\dim(\cW_{k,\cS})$ from Algorithm~O, we provide an explicit choice of nodes to exclude. We show that we can still perform the repair despite not downloading from those nodes (Corollary~\ref{cor:nodes_to_download}) and provide the explicit repair procedure (see Algorithm X).
    \item \textbf{Bandwidth Guarantee.} We prove that our trace-repair scheme has bandwidth at most $kt\log|\BB|$ bits (Theorem~\ref{thm:outperform}). This shows that the trace-repair scheme is universally as efficient as, or strictly better than the classical scheme in bandwidth, for all $k$ and finite field $\FF$. Furthermore, our optimized scheme consistently matches or outperforms existing state-of-the-art schemes (see Fig~\ref{fig:algorithm_comparison_vary_k} and Table~\ref{tab:bandwidth_comparison}).
\end{enumerate}

\begin{figure*}[h!]
    \centering

    \begin{subfigure}[b]{0.46\textwidth}
        \centering
        \includegraphics[width=\linewidth]{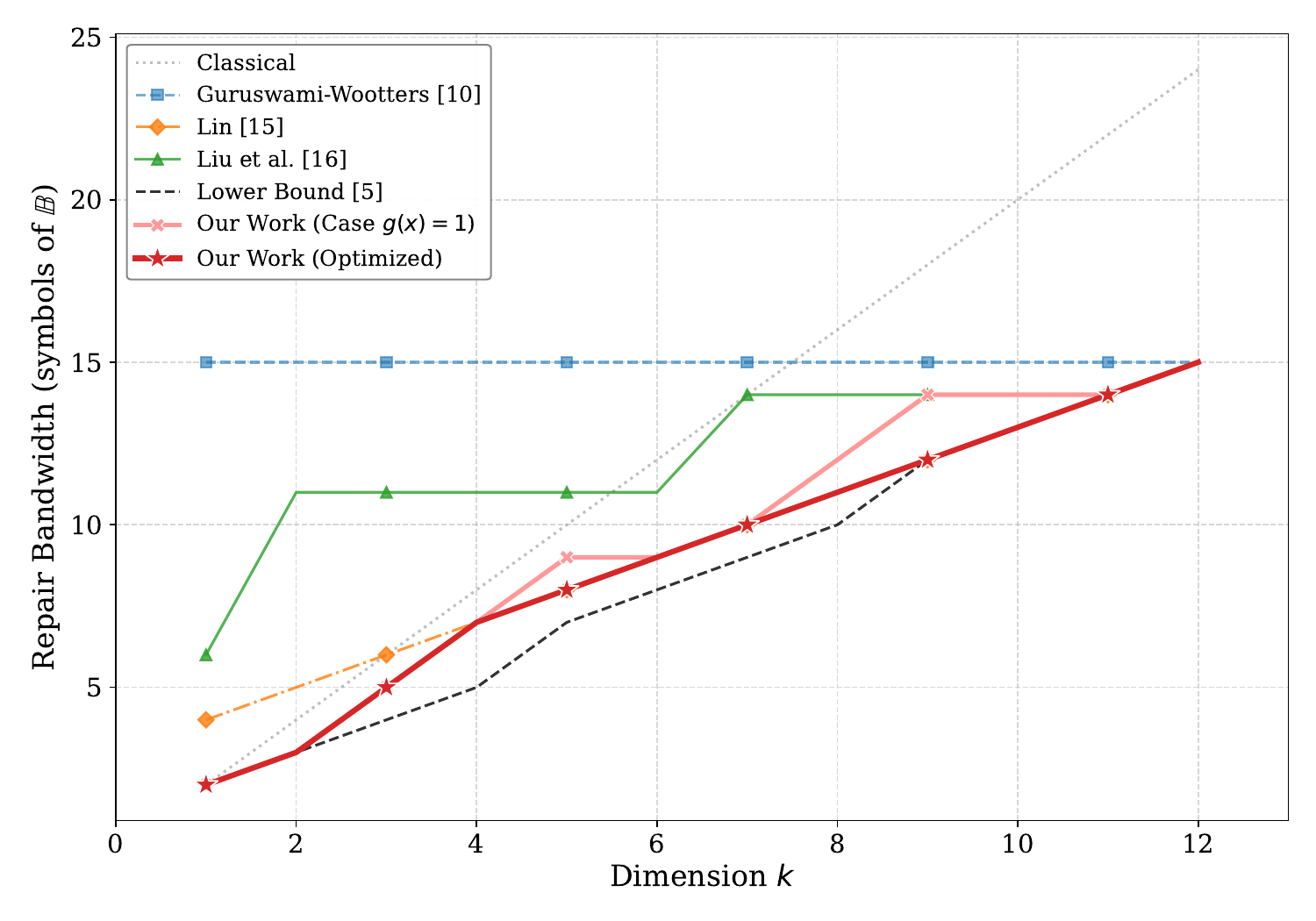}
        \caption{$\FF = \GF(4^2),\BB = \GF(4)$.}
        \label{fig:bandwidth_comparison_4^2}
    \end{subfigure}
    \hfill
    \begin{subfigure}[b]{0.46\textwidth}
        \centering
        \includegraphics[width=\linewidth]{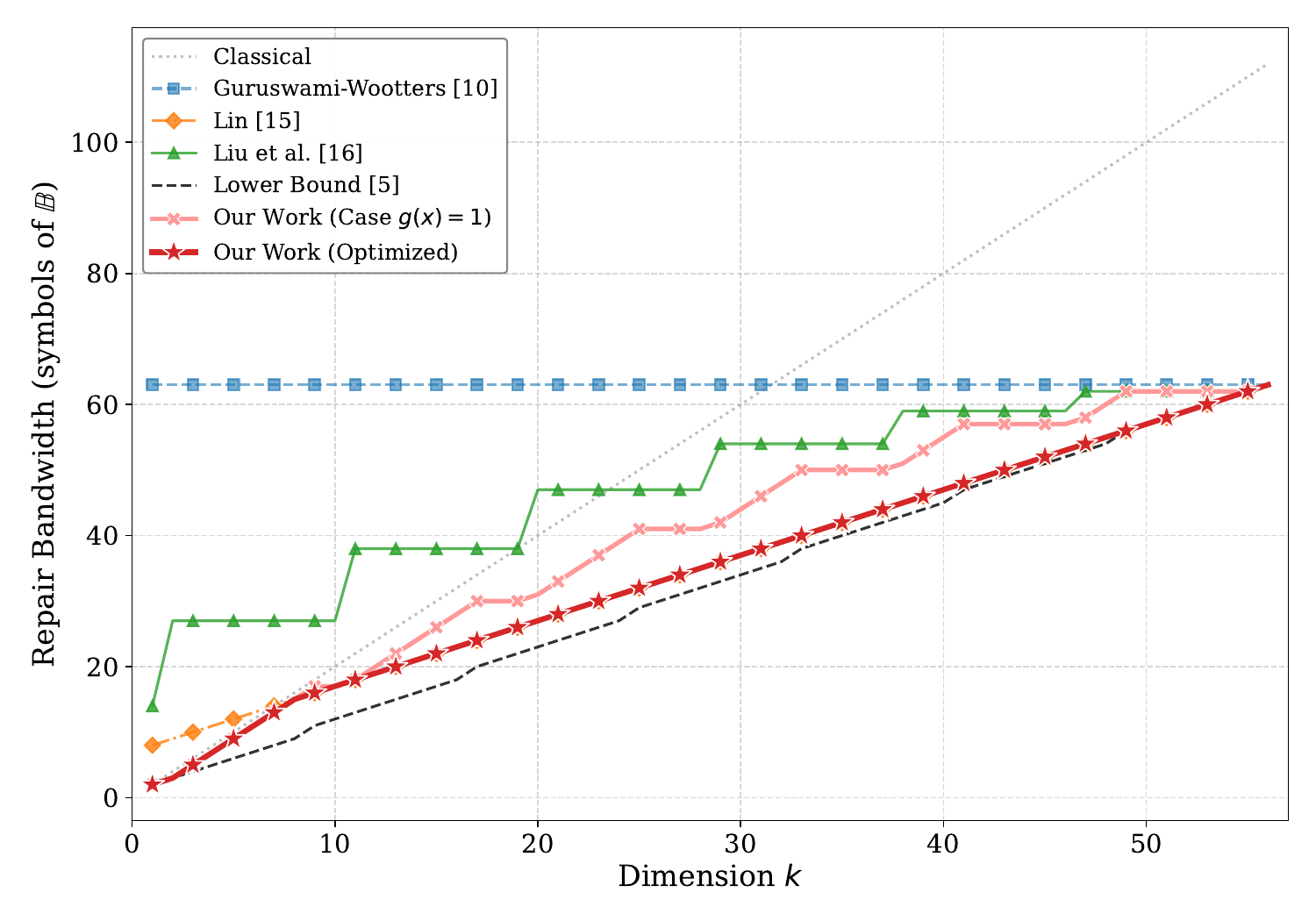}
        \caption{$\FF = \GF(8^2),\BB = \GF(8)$.}
        \label{fig:bandwidth_comparison_8^2}
    \end{subfigure}
    \vspace{3mm}
    
    \begin{subfigure}[b]{\textwidth}
        \centering
        \includegraphics[width = \linewidth]{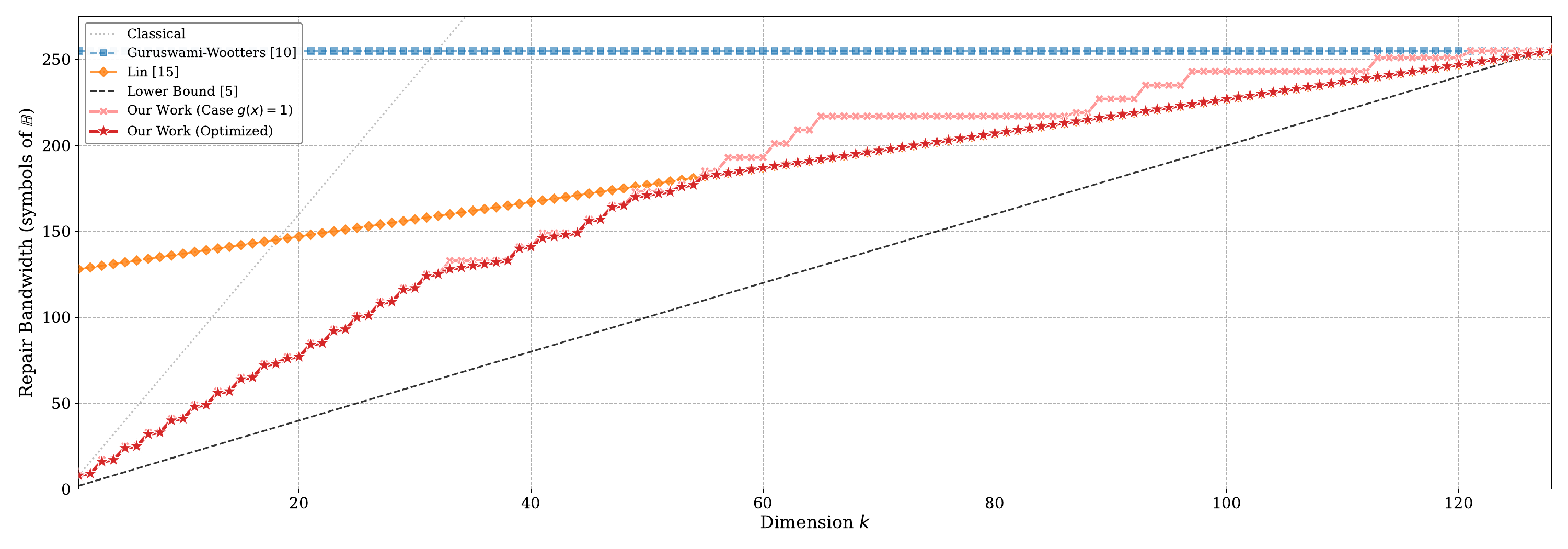}
        \caption{$\FF = \GF(2^8),\BB = \GF(2)$}
        \label{fig:bandwidth_comparison_2^8}
    \end{subfigure}
    \caption{The plots illustrate the repair bandwidth (in base field symbols) required for varying $k$. This shows empirically that our scheme never requires more bandwidth than the classical scheme (gray circles) or the schemes by Lin \cite{Lin2023} and Liu et al. \cite{Liu2024}.}
    \label{fig:algorithm_comparison_vary_k}
\end{figure*}

\begin{table}[h!]
    \centering
    \small
    \caption{Comparison of Repair Bandwidth (in base field symbols $\BB$) when $\FF = \GF(2^8)$ and $\BB = \GF(2)$ for $k=1$ to $k=54$. For $k\ge 55$, our optimized scheme has bandwidth that coincides with Lin~\cite{Lin2023}.}
    \label{tab:bandwidth_comparison}
    
    \begin{tabular}{cc}
        \begin{tabular}[t]{c c c c c c c}
            \toprule
             & \multicolumn{3}{c}{\textbf{Existing Schemes}} & \multicolumn{2}{c}{\textbf{Our Work}} & \\
            \cmidrule(lr){2-4} \cmidrule(lr){5-6}
            $k$ & Clas. & GW & Lin & $g(x)=1$ & Opt. & LB \\
            \midrule
            1 & 8 & 255 & 128 & 8 & 8 & 2 \\
            2 & 16 & 255 & 129 & 9 & \textbf{9} & 4 \\
            3 & 24 & 255 & 130 & 17 & \textbf{16} & 6 \\
            4 & 32 & 255 & 131 & 17 & \textbf{17} & 8 \\
            5 & 40 & 255 & 132 & 25 & \textbf{24} & 10 \\
            6 & 48 & 255 & 133 & 25 & \textbf{25} & 12 \\
            7 & 56 & 255 & 134 & 33 & \textbf{32} & 14 \\
            8 & 64 & 255 & 135 & 33 & \textbf{33} & 16 \\
            9 & 72 & 255 & 136 & 41 & \textbf{40} & 18 \\
            10 & 80 & 255 & 137 & 41 & \textbf{41} & 20 \\
            11 & 88 & 255 & 138 & 49 & \textbf{48} & 22 \\
            12 & 96 & 255 & 139 & 49 & \textbf{49} & 24 \\
            13 & 104 & 255 & 140 & 57 & \textbf{56} & 26 \\
            14 & 112 & 255 & 141 & 57 & \textbf{57} & 28 \\
            15 & 120 & 255 & 142 & 65 & \textbf{64} & 30 \\
            16 & 128 & 255 & 143 & 65 & \textbf{65} & 32 \\
            17 & 136 & 255 & 144 & 73 & \textbf{72} & 34 \\
            18 & 144 & 255 & 145 & 73 & \textbf{73} & 36 \\
            19 & 152 & 255 & 146 & 77 & \textbf{76} & 38 \\
            20 & 160 & 255 & 147 & 77 & \textbf{77} & 40 \\
            21 & 168 & 255 & 148 & 85 & \textbf{84} & 42 \\
            22 & 176 & 255 & 149 & 85 & \textbf{85} & 44 \\
            23 & 184 & 255 & 150 & 93 & \textbf{92} & 46 \\
            24 & 192 & 255 & 151 & 93 & \textbf{93} & 48 \\
            25 & 200 & 255 & 152 & 101 & \textbf{100} & 50 \\
            26 & 208 & 255 & 153 & 101 & \textbf{101} & 52 \\
            27 & 216 & 255 & 154 & 109 & \textbf{108} & 54 \\
            \bottomrule
        \end{tabular} &
        
        \begin{tabular}[t]{c c c c c c c}
            \toprule
             & \multicolumn{3}{c}{\textbf{Existing Schemes}} & \multicolumn{2}{c}{\textbf{Our Work}} & \\
            \cmidrule(lr){2-4} \cmidrule(lr){5-6}
            $k$ & Clas. & GW & Lin & $g(x)=1$ & Opt. & LB \\
            \midrule
            28 & 224 & 255 & 155 & 109 & \textbf{109} & 56 \\
            29 & 232 & 255 & 156 & 117 & \textbf{116} & 58 \\
            30 & 240 & 255 & 157 & 117 & \textbf{117} & 60 \\
            31 & 248 & 255 & 158 & 125 & \textbf{124} & 62 \\
            32 & 256 & 255 & 159 & 125 & \textbf{125} & 64 \\
            33 & 264 & 255 & 160 & 133 & \textbf{128} & 66 \\
            34 & 272 & 255 & 161 & 133 & \textbf{129} & 68 \\
            35 & 280 & 255 & 162 & 133 & \textbf{130} & 70 \\
            36 & 288 & 255 & 163 & 133 & \textbf{131} & 72 \\
            37 & 296 & 255 & 164 & 133 & \textbf{132} & 74 \\
            38 & 304 & 255 & 165 & 133 & \textbf{133} & 76 \\
            39 & 312 & 255 & 166 & 141 & \textbf{140} & 78 \\
            40 & 320 & 255 & 167 & 141 & \textbf{141} & 80 \\
            41 & 328 & 255 & 168 & 149 & \textbf{146} & 82 \\
            42 & 336 & 255 & 169 & 149 & \textbf{147} & 84 \\
            43 & 344 & 255 & 170 & 149 & \textbf{148} & 86 \\
            44 & 352 & 255 & 171 & 149 & \textbf{149} & 88 \\
            45 & 360 & 255 & 172 & 157 & \textbf{156} & 90 \\
            46 & 368 & 255 & 173 & 157 & \textbf{157} & 92 \\
            47 & 376 & 255 & 174 & 165 & \textbf{164} & 94 \\
            48 & 384 & 255 & 175 & 165 & \textbf{165} & 96 \\
            49 & 392 & 255 & 176 & 173 & \textbf{170} & 98 \\
            50 & 400 & 255 & 177 & 173 & \textbf{171} & 100 \\
            51 & 408 & 255 & 178 & 173 & \textbf{172} & 102 \\
            52 & 416 & 255 & 179 & 173 & \textbf{173} & 104 \\
            53 & 424 & 255 & 180 & 177 & \textbf{176} & 106 \\
            54 & 432 & 255 & 181 & 177 & \textbf{177} & 108 \\
            \bottomrule
        \end{tabular}
    \end{tabular}
\end{table}
\section{Preliminaries}\label{sec:prelim}

Let $[n]$ denote the set $\{1,2,\ldots,n\}$ and $[a,b]$ denote the set $\{a,a+1,\ldots, b\}$. Let $\BB$ be the finite field of size $q = p^m$ for some prime $p$ and let $\FF$ be its extension field of degree $t\ge 1$. Let $\FF^* = \FF\setminus\{0\}$. Let $\{u_1,\ldots,u_t\}$ be a basis of $\FF$ over $\BB$. We use $\FF[x]$ to denote the ring of polynomials over the finite field $\FF$.

We denote the {\em dual} of the code $\cC$ by $\cC^\perp$, and so, 
for each $\boldsymbol{c} = (c_1,\ldots,c_n)\in\cC$ and $\boldsymbol{c}^{\perp} = (c_1^\perp,\ldots,c_n^\perp)\in\cC^\perp$, it holds that $\sum_{i=1}^nc_ic_i^\perp = 0$. In this work, we focus on the ubiquitous Reed-Solomon code.

\begin{definition}
    The \textit{Reed-Solomon} code $\rs(\cA,k)$ over finite field $\FF$ of dimension $k$ with evaluation points $\cA\subseteq\FF$ is defined as
    $$
    \rs(\cA,k)\triangleq \{(c(\alpha))_{\alpha\in\cA}:c\in\FF[x],\deg(c(x))\le k-1\},
    $$
    while the \textit{generalized Reed-Solomon code} $\grs(\cA,k,\boldsymbol{\lambda})$ of dimension $k$ with evaluation points $\cA\subseteq \FF$ and multiplier vector $\boldsymbol{\lambda}\in(\FF\setminus\{0\})^n$ is defined as:
    $$
    \grs(\cA,k,\boldsymbol{\lambda}) \triangleq \{(\lambda_\alpha r(\alpha))_{\alpha\in\cA}:r\in\FF[x],\deg(r)\le n-k-1\}.
    $$
\end{definition}
\noindent It is well known (see~\cite{roth2006}) that the dual of $\rs(\cA,k)$ is $\grs(\cA,|\cA|-k,\vlambda)$ for some $\vlambda=(\lambda_{\alpha})_{\alpha\in \cA}$.
Furthermore, when $\cA = \FF$, we have $\lambda_\alpha = 1$ for all $\alpha\in\cA$.

Given a codeword $(c(\alpha))_{\alpha\in\cA}\in\rs(\cA,k)$, we can repair any $n-k$ erased code symbols by downloading any $k$ available code symbols. Guruswami and Wootters \cite{guruswamiwooters2017} proposed a repair scheme for a single erased code symbol which utilizes all $|\cA|-1$ available code symbols but downloads only parts of the information, resulting in a bandwidth of $(|\cA|-1)\log|\BB|$ bits. This improves the classical scheme for $k>(|\cA|-1)/t$. A lower bound for the repair bandwidth was established in \cite{Dau2017} and efforts to approach this lower bound were done by Lin \cite{Lin2023} and Liu et al. \cite{Liu2024}. We summarize the lower bound for the repair bandwidth in Theorem~\ref{thm:lower_bound}. 
\begin{theorem}[Dau and Milenkovic \cite{Dau2017}]\label{thm:lower_bound}
    Let $\cA \subseteq \FF$ and $n = |\cA|$. Any linear repair scheme for Reed-Solomon codes $\rs(\cA,k)$ over the extension field $\FF = \GF(q^t)$ that uses the subfield
$\BB= \GF(q)$ requires a bandwidth of at least
$
\ell \lfloor b_{AVE}\rfloor + (n-1-\ell) \lceil b_{AVE} \rceil,
$
where $b_{AVE}$ and $\ell$ defined as,
$$
b_{AVE} \triangleq \log_q \left(\frac{(n-1)|\FF|}{(n-k-1)(|\FF|-1)+(n-1)}\right),
$$
and
$$
\ell \triangleq \begin{cases}
    n-1,&\text{if } b_{AVE}\in\mathbb{Z},\\
    \left\lfloor\frac{L - (n-1)q^{-\lfloor b_{AVE}\rfloor}}{q^{-\lfloor b_{AVE}\rfloor} - q^{-\lceil b_{AVE}\rceil}}\right\rfloor, &\text{otherwise}.
\end{cases}
$$
Here, $L \triangleq ((n-k-1)(|\FF|-1)+(n-1))/|\FF|$.
\end{theorem}

{\color{black} 
\subsection{Repair Scheme by Guruswami and Wootters.}
Guruswami and Wootters \cite{guruswamiwooters2017} proposed a repair scheme as follows. For each $i\in[t]$, we consider a parity-check polynomial
$$
r_i(x) = \tr(u_i(x-\alpha^*))/(x-\alpha^*).
$$
Here, $r_i(x)$ is a polynomial of degree $q^{t-1}-1$. So, if $k \le |\cA| - q^{t-1}$, the following parity check equations hold for all $i\in[t]$:
\begin{align*}
    \sum_{\alpha\in\cA} r_i(\alpha)\lambda_\alpha c(\alpha) = 0 \implies
     \tr(u_i\lambda_{\alpha^*}c(\alpha^*)) = -\sum_{\alpha\in\cA\setminus\{\alpha^*\}} \tr(u_i(\alpha-\alpha^*))\tr\left(\frac{\lambda_\alpha c(\alpha)}{\alpha-\alpha^*}\right).
\end{align*}
Notice that, to we can compute $\tr(u_i \lambda_{\alpha^*}c(\alpha^*))$ for all $i\in[t]$ by downloading $\tr(\lambda_\alpha c(\alpha)/(\alpha-\alpha^*))\in\BB$ from all available nodes. Then, we compute
$$
f(\alpha^*) = \lambda_{\alpha^*}^{-1}\sum_{i\in[t]} u_i^\perp \tr(u_i \lambda_{\alpha^*} c(\alpha^*)).
$$
This results in a bandwidth of $(|\cA| - 1)\log|\BB|$ bits.}

{\color{black}
\subsection{Repair Scheme by Lin.}
Lin \cite{Lin2023} proposed a repair scheme for an erased Reed-Solomon code symbol without involving all available nodes. Lin realized that, given $k$, there is a gap between $|\cA|-k-1$ and the degree of $r_i(x)$ in the Guruswami-Wootters scheme, which can be utilized. Given $k$, for each $i\in[t]$, Lin considered a subset $\cS\subset \cA$ with $|\cS| = |\cA| - k - q^{t-1}$ and a different parity-check polynomial
$$
r_{i,\cS}(x) = g_{\cS}(x)\tr(u_i(x-\alpha^*))/(x-\alpha^*),\quad\text{where}\quad g_{\cS}(x) = \prod_{\beta\in \cS }(x-\beta) 
$$
Such a parity-check polynomial has a degree exactly $|\cA|-k-1$ and it has at least $|\cS|$ roots. Then, if $k\le |\cA|-q^{t-1}$, the following parity check equations hold for all $i\in[t]$:
\begin{align*}
    \sum_{\alpha\in\cA\setminus\cS} r_{i,\cS}(\alpha)\lambda_{\alpha}c(\alpha) = 0 \implies \tr(u_i\lambda_{\alpha^*}c(\alpha^*)) = -\sum_{\alpha\in\cA\setminus(\cS\cup \{\alpha^*\})} \tr(u_i(\alpha-\alpha^*))\tr\left(\frac{g_{\cS}(\alpha)\lambda_\alpha c(\alpha)}{\alpha-\alpha^*}\right)
\end{align*}
Here, we can ignore the code symbols $(c(\alpha))_{\alpha\in\cS}$ since $r_{i,\cS}(\alpha) = 0$ for all $\alpha\in\cS$. Thus, we repair $c(\alpha^*)$ utilizing only $|\cA|-|\cS|-1$ nodes, resulting in a bandwidth of $(k+q^{t-1}-1)\log|\BB|$ bits. It can be shown that, for $|\cA|-q^{t-1}-q+1\le k \le |\cA|-q^{t-1}$, the repair scheme bandwidth by Lin \cite{Lin2023} coincides with the lower bound in \cite{Dau2017}. 
}

\subsection{Repair Scheme by Liu et al.}
Recently, Liu et al. \cite{Liu2024} also proposed a repair scheme for an erased Reed-Solomon code symbol without involving all available nodes. In the special case for the trace repair scheme,
there exists a set $\cI \subseteq \cA$, so that, by downloading $\tr(\lambda_\alpha c(\alpha)/(\alpha-\alpha^*))$ for all $\alpha\in\cA\setminus(\cI\cup\{\alpha^*\} )$, we can recover $\tr(\lambda_\alpha c(\alpha)/(\alpha-\alpha^*))$ for all $\alpha\in \cI$. It can be done by forming new parity-check polynomials to repair traces that we do not download. After we obtain all required traces, we apply the Guruswami-Wootters scheme to repair the erased node. This results in a bandwidth of $(|\cA|-|\cI|-1)\log|\BB|$ bits. In what follows, for simplicity and without loss of generality, we assume that $c(0)$ is erased and we restate this special case of Liu et al. \cite{Liu2024} in Theorem~\ref{thm:liu}.


\begin{theorem}[Liu et al. \cite{Liu2024}]\label{thm:liu}
    Let $n= |\cA|$ and fix $k$. Let
    $$
        {\cal Y} = \{(0,y_1,\ldots,y_{n-1}): y_i \in\frac{1}{\alpha_i}\BB\}
    $$
    and
    $$
    {\cal W}_k = \rs(\cA,k)^\perp\cap\cY.
    $$
    If $\dim({\cal W}_k) \ge d$, then we can repair $c(0)$ with bandwidth of $(n-d-1)\log|\BB|$ bits.
\end{theorem}

In the same work, Liu et al. provided lower bounds for $\dim(\cW_k)$ in two settings: namely, $\FF=\GF(p^2)$ with $\BB=\GF(p)$; and $\FF=\GF(2^s)$ with $\BB=\GF(2^{s/2})$ for even $s\ge 2$. In this work, we study arbitrary finite fields and, crucially, determine the exact value of $\dim(\cW_k)$.

\section{Dimension Analysis}\label{sec:our_work}

Let us first formulate our problem. 
Let $\cA=\FF = \GF(q^t)$ for some prime power $q$ and $n=|\FF|$. 
Let $\omega$ be the primitive element of $\FF$.
Consider the codeword $(c(\alpha))_{\alpha\in\cA}\in\rs(\cA,k)$ where $c(0)$ is erased. 
Let $\cW_{k,\cS}$ be as defined in Theorem~\ref{thm:main_theorem}. 
Our goal is to determine $\dim(\cW_{k,\cS})$ exactly, and to identify the helper nodes to download from.

This section is organized as follows. In Section~\ref{sec:main_idea}, we present the core concept of our scheme by introducing \textit{Trace Repair Compatible Polynomial} in Definition~\ref{def:trace-repair-compatible}. Such polynomials must satisfy two conditions: the Base Field Condition and the Degree Condition. We establish that the repair bandwidth is directly determined by the dimension of the space formed by these polynomials. Next, in Section~\ref{sec:basis_of_F}, we analyze the structure of polynomials satisfying the Base Field Condition, that is, those mapping $\FF$ to $\BB$. We denote this space by $\cF\triangleq \{f\text{ polynomial } | f:\FF\to \BB\}$ and determine its $\BB$-basis. Finally, in Section~\ref{sec:basis_of_W}, we apply the Degree Condition to this basis to explicitly identify the subspace of Trace-Repair Compatible polynomials. Consequently, determining $\dim(\cW_{k,\cS})$ is reduced to finding the dimension of this subspace.

\subsection{Main Idea}\label{sec:main_idea}
In this section, we discuss how to fully utilize the previous work's idea, resulting in a further bandwidth reduction. To do so, we introduce a class of polynomials essential to our proposed scheme.

\begin{definition}[Trace-Repair Compatible Polynomial]\label{def:trace-repair-compatible}
    Fix $\cS\subset\FF$ and let $g(x) = \prod_{\beta\in\cS}(x-\beta)$ if $\cS$ is nonempty; $g(x) = 1$ if $\cS = \emptyset$. A polynomial $T(x)\in\FF[x]$ is called \textit{Trace Repair Compatible} polynomial with respect to $g(x)$ if it satisfies the following two conditions:
    \begin{enumerate}
        \item (Base Field Condition) $T:\FF\to\BB$, and
        \item (Degree Condition) The polynomial $h(x) = g(x) T(x)/x$ satisfies $h(0) = 0$ and $\deg(h(x))\le n-k-1$.
    \end{enumerate}
\end{definition}
\noindent There is a strong connection between Trace-Repair Compatible polynomials and the subspace $\cW_{k,\cS}$. 
Specifically, let $\cV_{k,\cS}$ be the vector space of all Trace-Repair Compatible polynomials with respect to $g(x)$, and let $\cT_{k,\cS}=\{T_1,\ldots,T_d\}$ be a basis for $\cV_{k,\cS}$. 
In other words, $\cV_{k,\cS} = \vspan_{\BB} \cT_{k,\cS}$.
The dimension of this space determines the repair bandwidth.
\begin{theorem}\label{thm:dim_V}
    Let $\cT_{k,\cS} = \{T_1,\ldots,T_d\}$ be a basis of $\cV_{k,\cS}$. Then, $\dim(\cW_{k,\cS}) = \dim(\cV_{k,\cS}) = d$.
\end{theorem}
\begin{proof}
    Let $\phi:\cV_{k,\cS} \to \cW_{k,\cS}$ be the evaluation map defined by $\phi(f) = (g(\alpha)f(\alpha)/\alpha)_{\alpha\in \FF}$. Clearly, $\phi$ is a linear map.

    \noindent \textbf{Injectivity}: Suppose that $\phi(f) = \boldsymbol{0}$. Then, $h(\alpha) = g(\alpha)f(\alpha)/\alpha = 0$ for all $\alpha\in\FF$. Since $\deg(h(x))\le n-k-1<n$ and $h$ has $n$ roots, $h(x)$ must be a zero polynomial. Therefore, $f = 0$.
    
    \noindent \textbf{Surjectivity}: By definition, any $\boldsymbol{w}\in\cW_{k,\cS}$ corresponds to some polynomial $f$ satisfying the Trace Repair Compatible conditions. Thus $f\in\cV_{k,\cS}$.

    \noindent Therefore, $\phi$ is a linear isomorphism, and $\dim(\cW_{k,\cS}) = \dim(\cV_{k,\cS}) = d$.
\end{proof}

Theorem~\ref{thm:dim_V} establishes that the polynomial space $\cV_{k,\cS}$ provides a concrete representation of the subspace $\cW_{k,\cS}$. 
Specifically, finding a basis for $\cW_{k,\cS}$ is equivalent to finding a basis for $\cV_{k,\cS}$.
We now utilize this correspondence to prove the existence of a repair scheme with the bandwidth guaranteed in Theorem~\ref{thm:main_theorem}.

\begin{proof}[Proof of Theorem~\ref{thm:main_theorem}]
    Recall that, to repair $c(0)$, we are required to have $\tr(g(\alpha)c(\alpha)/\alpha)$ for all $\alpha\in\FF^*$. Let $\cS\subset \FF^*$ be given and set $g(x) = \prod_{\beta\in\cS}(x-\beta)$ if $\cS$ is nonempty, and $g(x) = 1$ if $\cS = \emptyset$. Suppose that $\dim(\cW_{k,\cS}) = \dim(\cV_{k,\cS}) = d$ (by Theorem~\ref{thm:dim_V}) and let $\cT_{k,\cS}$ be a basis for $\cV_{k,\cS}$. For each $i\in[d]$, let $h_i(x) = g(x)T_i(x)/x$ and the following parity-check equations hold:
    $$
    \sum_{\alpha\in\FF^*} h_i(\alpha)c(\alpha) = 0 \implies \sum_{\alpha \in \FF^*} T_i(\alpha) \tr(g(\alpha)c(\alpha)/\alpha) = 0
    $$
    Since $g(\alpha) = 0$ for all $\alpha \in \cS$, the terms in the summation corresponding to $\alpha\in\cS$ vanish. Hence, we have
    $$
    \sum_{\alpha \in \FF^*\setminus\cS} T_i(\alpha) \tr(g(\alpha)c(\alpha)/\alpha) = 0.
    $$
    In other words, we can exclude these $|\cS|$ nodes in the repair scheme. Furthermore, in matrix form we have, $\boldsymbol{Mx} = \boldsymbol{0}$
    where $\boldsymbol{M}$ is a $d\times(n-1-|\cS|)$ matrix with entries $M_{j,\alpha} = T_j(\alpha)$ and $\boldsymbol{x}$ is the vector of traces. Since $\boldsymbol{M}$ has a row rank $d$, there exists $\cI\subseteq \FF^*\setminus\cS$ of size $d$, such that $\boldsymbol{M}_{\cI}$ is invertible. Hence, we can recover traces corresponding to $\cI$ using traces corresponding to $\FF^*\setminus(\cS\cup \cI)$. Therefore, to repair $c(0)$, we are only required to download a trace in $\BB$ from $|\FF^*|-|\cS| - |\cI| = n - 1 - (|\cS| + \dim(\cW_{k,\cS}))$ nodes.
\end{proof}

\begin{remark}
    The proof above relies on the matrix rank property to guarantee the existence of $\cI$ of size $d = \dim(\cV_{k,\cS})$. To establish an explicit scheme, we are left to determine the exact value of $d$ (in Theorem~\ref{thm:dim_of_WkS}) and provide explicit helper nodes to perform the repair (in Corollary~\ref{cor:nodes_to_download}).
\end{remark}

\subsection{Base Field Condition: The Polynomials $f:\FF\to\BB$.}\label{sec:basis_of_F}

In this section, we construct an explicit $\BB$-basis of $\cF = \{f\text{ polynomial }| f:\FF\to\BB\}$ by analyzing the effect on the condition to its coefficients and exponents. We begin by recalling the following terminology, the \textit{cyclotomic cosets}.
\begin{definition}
    Fix $t$ and $q = p^m$. A subset $\{a_1,\ldots,a_s\} \subset \{0,1,\ldots,q^t-2\}$ is called a \textit{cyclotomic coset} if $qa_j = a_{j+1}(\text{mod } q^t-1)$ for all $j\in\{1,\ldots, s-1\}$ and $qa_s = a_1(\text{mod } q^t-1)$. 
    The collection of all such cosets partitions $\{0,1,\ldots,q^t-2\}$ and we refer to it as the \textit{collection of cyclotomic cosets modulo $q^t-1$}.
\end{definition}
In this work, we use
$
C_i =\{a^{(i)},a^{(i)}q,\ldots,a^{(i)}q^{s_i-1}\}
$
to denote the $i$-th cyclotomic coset in its collection $\Xi$ and $s_i = |C_i|$. It is clear that $s_i\le t$ as $a^{(i)}q^t = a^{(i)}(\mathrm{mod}\,{q^t-1})$. For example, suppose $q = 3$ and $t=2$. Then, the collection of cyclotomic cosets of $\{0,1,\ldots, 7\}$ is $\{\{0\},\{1,3\},\{2,6\},\{4\},\{5,7\}\}$. The size of each coset is at most $2$.

There is a direct correspondence between the cyclotomic cosets and polynomials in $\cF$. A polynomial $f(x)\in\FF[x]$ maps $\FF$ to $\BB$ if and only if $[f(x)]^q = f(x)$ for all $x\in\FF$. This condition forces the coefficients and expoonents within the same cyclotomic coset to be related. We formalize this finding in Lemma~\ref{lem:characterizing_f}.

\begin{lemma}\label{lem:characterizing_f}
Let
$
f(x) = \sum_{i=0}^{q^t-2} f_ix^i
$. Then, $f\in\cF$ if and only if
$f(x) = \sum_{i=1}^{|\Xi|}\sum_{j=0}^{s_i-1}f_{a^{(i)}}^{q^j}x^{a^{(i)}q^j}$ with $f_0\in\BB$.
\end{lemma} 
\begin{proof}

Note that, $f\in\cF$ if and only if $[f(x)]^q = f(x)$ for all $x\in\FF$. That is,
$$
\sum_{i=0}^{q^t-2} f_i^q x^{iq} = \sum_{i_*=0}^{q^t-2} f_{i_*}x^{i_*}\implies f_{iq} = f_i^q,
$$
for all $i\in[0,q^t-2]$. From the above, we must have $f_0\in\BB$. Then, note that $\{iq : i\in[0,q^t-2]\} = [0,q^t-2]$ can be partitioned into $C_1,\ldots, C_{|\Xi|}$. Therefore, splitting the summation according to $C_i$ yields,
$$
f(x) = \sum_{i=1}^{|\Xi|}\sum_{j=0}^{s_i-1} f_{a^{(i)}q^j}x^{a^{(i)}q^j} = \sum_{i=1}^{|\Xi|}\sum_{j=0}^{s_i-1} f_{a^{(i)}}^{q^j}x^{a^{(i)}q^j}.
$$
\end{proof}
\begin{example}
    Suppose $\FF = \GF(2^3),\BB = \GF(2)$. Then, we have $\Xi =\{\{0\},\{1,2,4\},\{3,6,5\}\}$. Let $f(x) = \sum_{i=0}^6 f_i x^i$. Then, comparing the coefficients of $f(x)$ and
    $$
    [f(x)]^2 = f_0^2 + f_1^2 x^2 + f_2^2 x^4 + f_3^2 x^6 + f_4^2 x + f_5^2 x^3 + f_6 x^5, 
    $$
    yields $f_0\in\BB$, $f_2 = f_1^2, f_4 = f_1^4, f_6 = f_3^2, $ and $f_5 = f_3^4$. Hence,
    $$
    f(x) = f_0 + (f_1 x + f_1^2 x^2 + f_1^4 x^4) + (f_3 x^3 + f_3^2 x^6 + f_3^4 x^5),
    $$
    with $f_0\in\BB$, which expression is partitioned according to the cyclotomic cosets in $\Xi$.\qed
\end{example}

Lemma~\ref{lem:characterizing_f} shows a crucial structural property, that is, the space $\cF = \{f:f:\FF\to\BB\}$ decomposes into disjoint subspaces corresponding to each cyclotomic coset. In what follows, we let $\cF \triangleq \bigoplus_{i} \cF_i$, where $\cF_i$ is the subspace of $\cF$ with polynomials corresponding to $C_i$. Furthermore, for each cyclotomic coset $C_i$, the corresponding polynomial exponents are the entries in $C_i$ and its coefficients are fully determined by the leading coefficient $f_{a^{(i)}}$. Since the leading coefficient $f_{a^{(i)}}$ lies in the extended field $\FF$, it can be expressed as a $\BB$-linear combination of $\BB$-basis elements of $\FF$. Let this $\BB$-basis be $\{1,\omega,\ldots,\omega^{t-1}\}$. With this expression, we can construct an explicit $\BB$-basis for polynomials corresponding to $C_i$.

To be more precise, fix $C_i$ and let $f_{a^{(i)}} = \sum_{\ell=0}^{t-1} f_{a^{(i)}}^{(\ell)} \omega^\ell$ for some $f_{a^{(i)}}^{(0)},\ldots,f_{a^{(i)}}^{(t-1)}\in\BB$. Then, the corresponding polynomial to this cyclotomic coset $C_i$ can be expressed as
$$
\sum_{j=0}^{s_i - 1} f_{a^{(i)}}^{q^j} x^{a^{(i)}q^j}  
= \sum_{j=0}^{s_i-1} \left(\sum_{\ell = 0}^{t-1} f_{a^{(i)}}^{(\ell)}\omega^\ell\right)^{q^j} x^{a^{(i)}q^j}
= \sum_{j=0}^{s_i-1} \sum_{\ell = 0}^{t-1} f_{a^{(i)}}^{(\ell)}\omega^{\ell{q^j}} x^{a^{(i)}q^j}
= \sum_{\ell=0}^{t-1}f_{a^{(i)}}^{(\ell)}\left(\sum_{j=0}^{s_i-1} \omega^{\ell q^j}x^{a^{(i)}q^j}\right).
$$
This shows that we can express the polynomial as $t$ $\BB$-linear combination of polynomials $T_0^{(i)},\ldots,T_{t-1}^{(i)}$ where
$$
T_\ell^{(i)}(x) \triangleq \sum_{j=0}^{s_i-1} \omega^{\ell q^j}x^{a^{(i)}q^j},\quad\text{for }\ell\in[0,t-1].
$$

However, we observe a dimensional reduction, that is, we need not consider all $t$ of them but only $s_i$ of them, and $s_i \le t$. Let us demonstrate this observation using Example~\ref{ex:dim_reduction} and formalize it in Theorem~\ref{lem:basis}.

\begin{example}\label{ex:dim_reduction}
    Let $\FF = \GF(2^4)$ and $\BB = \GF(2)$, and we consider a polynomial $f:\FF\to\BB$ with exponents in the cyclotomic coset $\{5,10\}$. Here, $s_i = 2$ and $t= 4$. From Lemma~\ref{lem:characterizing_f}, we know that it takes the form $f(x) = f_5 x^5 + f_5^2 x^{10}$. Since $f_5\in\FF$, we can write $f_5 = f_5^{(0)} + f_5^{(1)} \omega + f_5^{(2)} \omega^2 + f_5^{(3)} \omega^3$ with each $f_5^{(\ell)}\in\BB$. Rewriting $f(x)$ with this representation, we have
$f(x) = f_5^{(0)}(x^5 + x^{10}) + f_5^{(1)}(\omega x^5 + \omega^2 x^5) + f_5^{(2)}(\omega^2 x^5 + \omega^4 x^5) + f_5^{(3)}(\omega^3 x^5 + \omega^6 x^5)$. In other words, $f(x)\in \text{span}_\BB \{x^5+x^{10}, \omega x^5 + \omega^2 x^{10}, \omega^2 x^5 + \omega^4 x^5, \omega^3 x^5 + \omega^6 x^5\}$. But, $\omega^2 x^5 + \omega^4 x^{10} = \lambda_1 (x^5+x^{10}) + \lambda_2 (\omega x^5 + \omega^2 x^{10})$ for some $\lambda_1,\lambda_2\in\BB$, that is
$$
\begin{bmatrix}
    \lambda_1\\
    \lambda_2
\end{bmatrix} = \begin{bmatrix}
    1 & \omega\\
    1 & \omega^2
\end{bmatrix}^{-1} \begin{bmatrix}
    \omega^2\\
    \omega^4
\end{bmatrix}.
$$
This also occurs for $\omega^3 x^5 +\omega^6 x^5$. Therefore, $f(x)\in\text{span}_{\BB}\{x^5 + x^{10}, \omega x^5 + \omega^2 x^{10}\}$.
\end{example}

\begin{lemma}\label{lem:basis}
    Fix a cyclotomic coset $C_i$ of size $s_i$. Let $\cF_i$ be the subspace of $\cF$ with polynomials corresponding to $C_i$. Then, $\{T_{\ell}^{(i)}\}_{\ell\in[0,s_i-1]}$ is a $\BB$-basis of $\cF_i$.
\end{lemma}
\begin{proof}
    First, we verify that $T_\ell^{(i)}\in\cF_i$ for all $\ell\in[0,s_i-1]$. This is clear since
    $$[T_{\ell}^{(i)}(x)]^q = \omega^{\ell q^{s_i}} x^{a^{(i)}q^{s_i}}+\sum_{j=1}^{s_i-1}\omega^{\ell q^j}x^{a^{(i)}q^j} = T_\ell^{(i)}(x).$$
    Next, we show that $\{T_\ell^{(i)}\}_{\ell\in[0,s_i-1]}$ is $\BB$-linearly independent and $\cF_i \subseteq \vspan_{\BB}\{T_\ell^{(i)}\}_{\ell\in[0,s_i-1]} $. To show linear independence, we suppose that there is $\lambda_0,\ldots,\lambda_{s_i-1}$ such that $\sum_{\ell=0}^{s_i-1}\lambda_\ell T_\ell^{(i)}(x) =0$. We show that $\lambda_0=\cdots=\lambda_{s_i-1}=0$. Indeed, by comparing coefficients and expressing in matrix form, we have $\boldsymbol{W}\boldsymbol{\lambda} = \boldsymbol{0}$ with $\boldsymbol{\lambda} = \begin{bmatrix}
        \lambda_0& \cdots & \lambda_{s_i-1}
    \end{bmatrix}^\top$ and
    $$
    \boldsymbol{W} = \begin{bmatrix}
    1 & \omega  & \cdots & \omega^{s_i-1}\\
    1 & \omega^q  & \cdots & (\omega^{q})^{s_i-1}\\
    \vdots & \vdots  &\ddots & \vdots\\
    1 & \omega^{q^{s_i-1}} & \cdots & (\omega^{q^{s_i-1}})^{s_i-1} 
\end{bmatrix}.
    $$
    Since $\boldsymbol{W}$ is a Vandermonde matrix, it is invertible and $\lambda_0 = \cdots = \lambda_{s_i-1}=0$. Next, we show that for any $f\in\cF_{i}$, $f\in\vspan_{\BB}\{T_{\ell}^{(i)}\}_{\ell\in[0,s_i-1]}$. It is sufficient to show that, for any $\ell^*\in[s_i,t-1]$, $f_{\ell^*}\in\vspan_{\BB}\{T_{\ell}^{(i)}\}_{\ell\in[0,s_i-1]}$, that is,
    $$
T_{\ell^*}^{(i)}(x) = \sum_{\ell=0}^{s_i-1} \lambda_\ell T_\ell^{(i)}(x)
$$
for some $\lambda_0,\ldots,\lambda_{s_i-1}\in\BB$. Indeed, in matrix form, $\boldsymbol{W\lambda} = [\omega^{\ell^*},\omega^{\ell^{*}q},\ldots,\omega^{\ell^* q^{s_i-1}}]^\top$. Hence, $\boldsymbol{\lambda}=\boldsymbol{W}^{-1}[\omega^{\ell^*},\omega^{\ell^{*}q},\ldots,\omega^{\ell^* q^{s_i-1}}]^\top$. Furthermore, note that $T_\ell^{(i)}(\alpha)\in\BB$ for all $\alpha\in\FF$, i.e., $[T_\ell^{(i)}(x)]^q = T_\ell^{(i)}(x)$. Therefore,
\begin{align*}
    T_{\ell^*}^{(i)}(x) = \sum_{\ell=0}^{s_i-1} \lambda_\ell T_\ell^{(i)}(x) \iff T_{\ell^*}^{(i)}(x) = \sum_{\ell=0}^{s_i-1} \lambda_\ell^q T_\ell^{(i)}(x).
\end{align*}
This implies $\lambda_\ell^q = \lambda_\ell\in\BB$ for all $\ell\in[0,s_i-1]$.
\end{proof}

Lemma~\ref{lem:basis} provides a basis for the subspace $\cF_i$ whose polynomials map from $\FF$ to $\BB$ and has exponents corresponding to the cyclotomic coset $C_i$. To construct a basis for the subspace $\cF$, we exploit the partition property of cyclotomic cosets. Specifically, since the collection $\Xi$ consists of pairwise disjoint cyclotomic cosets, any two polynomials corresponding to different cosets share no common exponents. Consequently, basis of $\cF_i$ and $\cF_j$ are $\BB$-linearly independent for any $i\ne j$, and $\cF = \bigoplus_i \cF_i$. Hence, we can use the union of disjoint bases of $\cF_i$ as a basis of $\cF$. We formalize this in Theorem~\ref{thm:dim_of_F}.

\begin{theorem}\label{thm:dim_of_F}
    Let $\cT\triangleq \left\{T_{\ell}^{(i)}:\ell\in[0,s_i-1], i\in[|\Xi|]\right\}$. Then, $\cT$ is a $\BB$-basis of $\cF$ and $\dim(\cF) = \sum_{i=0}^{|\Xi|} s_i$.
\end{theorem}
\begin{proof}
     Let $f$ be a corresponding polynomial of $\cF$. Due to Lemma~\ref{lem:characterizing_f},
    $$
    f(x) = \sum_{i=1}^{|\Xi|}\sum_{j=0}^{s_i-1}f_{a^{(i)}}^{q^j}x^{a^{(i)}q^j}.
    $$
Then, due to Lemma~\ref{lem:basis}, we can write
$$
f(x) = \sum_{i=1}^{|\Xi|}T^{(i)}(x) = \sum_{i=1}^{|\Xi|}\sum_{\ell=0}^{s_i-1} \lambda_\ell^{(i)} T_{\ell}^{(i)}(x).
$$
Again, due to Lemma~\ref{lem:basis} and since all distinct cyclotomic cosets in $\Xi$ are disjoint, $\{T_\ell^{(i)}:\ell\in[0,s_i-1],i\in[|\Xi|]\}$ is $\BB$-linearly independent. Then, the result follows.
\end{proof}
Theorem~\ref{thm:dim_of_F} fully characterizes the subspace $\cF$ by providing an explicit basis of it. In other words, it provides all possible candidate of $\BB$-linearly independent polynomial satisfying the Base Field Condition. Consequently, we are left to determine which of these candidates satisfy the Degree Condition (see Definition~\ref{def:trace-repair-compatible}). As a result, we obtain the maximum number of possible Trace Repair Compatible polynomials. To deliver the idea better, we first consider the case when $g(x) = 1$. With this choice of $g(x)$, $\cW_{k,\cS}$ is reduced to $\cW_k$ in Theorem~\ref{thm:liu}. In the next section, we show that we can find the dimension of $\cW_k$ exactly utilizing the $\BB$-basis of $\cF$ in Theorem~\ref{thm:dim_of_F}.

\subsection{The Exact Dimension of $\cW_k$ (Case $g(x) = 1$).}\label{sec:basis_of_W}
In this section, we find the exact dimension of $\cW_k$ which corresponds to the repair scheme by Liu et al. (see Theorem~\ref{thm:liu}). By Theorem~\ref{thm:dim_V}, this is equivalent to finding the dimension of the space of Trace Repair Compatible polynomials $\cV_k$ with $g(x) = 1$. The high-level idea is as follows. By having a $\BB$-basis of $\cF$ (see Theorem~\ref{thm:dim_of_F}), we have all possible $\BB$-linearly independent candidate polynomials satisfying the Base Field Condition in Definition~\ref{def:trace-repair-compatible}. Then, from the basis, we pick those satisfying the Degree Condition in Definition~\ref{def:trace-repair-compatible}. These chosen polynomials form a basis of $\cV_k$, and by Theorem~\ref{thm:dim_V}$, \dim(\cW_k) = \dim(\cV_k)$.

By Theorem~\ref{thm:dim_of_F}, any Trace Repair Compatible polynomial must lie in the span of $\cT = \{T_{\ell}^{(i)}:\ell\in\{0,s_i-1\},i\in[|\Xi|]\}$. However, for the polynomial to be Trace Repair Compatible, a basis element $T_{\ell}^{(i)}$ must also satisfy the Degree Condition. That is, for each $\ell\in\{0,s_i-1\}$ and $i\in[|\Xi|]$, let $h_{\ell}^{(i)} = T_{\ell}^{(i)}(x)/x$ and $h$ must satisfy (a) $h_\ell^{(i)}(0) = 0$, and (b) $\deg(h_\ell^{(i)})\le n-k-1$. Let us analyze them separately.
\begin{enumerate}[label = (\alph*)]
    \item Since $h_\ell^{(i)}(0) = 0$, then $T_{\ell}^{(i)}(x)$ cannot have the term $x$. In other words, we exclude those $T_{\ell}^{(i)}(x)$ that correspond to the coset containing $1$.
    \item For $k\ge 2$, we must have $\deg(T_\ell^{(i)}(x)) \le n-k$ with $0$ constant. In other words, we exclude those $T_\ell^{(i)}(x)$ that correspond to the cyclotomic coset containing $0$ and all cyclotomic cosets with some entry more than $n-k$.
\end{enumerate}
Based on this, we define the set of valid cosets 
$$\Xi_{k}^* \triangleq\begin{cases}
    \{C\in \Xi: 0\notin C, 1\notin C, \max(C)\le n-k\},&\text{for }k\ge 2,\\
    \{C\in \Xi: 1\notin C\},&\text{for }k= 1.
\end{cases} $$
We use $C_{i,k}^*$ to denote the $i$-th cyclotomic coset in $\Xi_k^*$ and let $s_{i,k}^* = |C_{i,k}^*|$. Hence, the basis of $\cV_k$ is simply the union of the basis polynomials corresponding to the cyclotomic cosets in $\Xi_k^*$. That is, 
$$
\cT_k^{*} \triangleq \left\{T_{\ell}^{(i)}:\ell\in[0,s_{i,k}^*-1], i\in[|\Xi_k^*|]\right\}
$$
\begin{example}
    Suppose $\FF = \GF(3^2),\BB = \GF(3),k=3$. We have $\Xi =\{\{0\},\{1,3\},\{2,6\},\{4\},\{5,7\}\}$ and $\Xi^*_k = \{\{2,6\},\{4\}\}$. Then, $\cT_k = \{T_0^{(1)}(x), T_1^{(1)}(x), T_0^{(2)}(x)\}$, where
    \begin{align*}
        T_0^{(1)}(x) &= x^2 + x^6,\\
        T_1^{(1)}(x) &= \omega x^2 + \omega^3 x^6,\\
        T_0^{(2)}(x) &= x^4.
    \end{align*}
\end{example}
\noindent We conclude this section with Theorem~\ref{thm:dim_of_Wk}.

\begin{theorem}\label{thm:dim_of_Wk}
    The dimension of the subspace $\cW_k$ is exactly the sum of the size of each valid cyclotomic coset:
    $$
    \dim(\cW_k) = \dim(\cV_k) = \sum_{i=1}^{|\Xi_k^*|} s_{i,k}^*.
    $$
\end{theorem}
\subsection{Further Improvement: Choice of $g(x)$ and Dimension of $\cW_{k,\cS}$.}\label{sec:opt}

In this section, we further lower the repair bandwidth by utilizing the idea proposed by Lin \cite{Lin2023}. Recall that, Lin introduced a polynomial $g(x) = \prod_{\beta\in\cS} (x-\beta)$ into the parity-check polynomial $r_{i,\cS}$. Specifically, Lin proposed $r_{i,\cS} = g(x) \tr(u_i x)/x$ which vanishes when $x=\beta\in\cS$. By using $r_{i,\cS}$, the term corresponding to $\beta\in\cS$ in the parity-check equations vanishes. This allows us to omit code symbols at $\cS$ entirely in the repair process. Lin fully utilized the degree gap between $\tr(u_i x)/x$ and $n-k-1$, that is, by setting $|\cS| = (n-k-1) - \deg(\tr(u_i x)/x) = n - k - q^{t-1}$. However, by doing so, it gives too big of restrictions on the subspace $\cW_{k,\cS}$ which leads to no room of improvement via the subspace method. Therefore, in this work, we relax the condition on $|\cS|$.

Furthermore, we observe a trade-off between $\cS$ and $\dim(\cW_{k,\cS})$. To see this, we discuss how to obtain $\dim(\cW_{k,\cS})$, which is very similar to the way we get $\dim(\cW_k)$. We check if a basis element $T_\ell^{(i)}\in\cT$ satisfies the Degree Condition in Definition~\ref{def:trace-repair-compatible}. That is, for each $\ell\in\{0,s_i-1\}$ and $i\in[|\Xi|]$, let $h_\ell^{(i)}(x) = g(x)T_\ell^{(i)}(x)/x$ and $h_\ell^{(i)}$ must satisfy (a) $h_\ell^{(i)}(0) = 0$, and (b) $\deg(h_\ell^{(i)}(x))\le n-k-1$. We analyze these separately.
\begin{enumerate}[label = (\alph*)]
    \item Note that, $g(0)$ is nonzero. So to ensure $h_\ell^{(i)}(0) = 0$, $T_\ell^{(i)}(x)$ cannot have the term $x$. In other words, we exclude those $T_\ell^{(i)}$ that correspond to the cyclotomic coset containing $1$.
    \item For $k\ge 2$, we must have $\deg(T_\ell^{(i)}(x))\le n-k-|\cS|$ with $0$ constant. In other words, we exclude those $T_\ell^{(i)}(x)$ that correspond to the cyclotomic coset containing $0$ and all cyclotomic cosets with some entry more than $n-k$.
\end{enumerate}
For a general $|\cS|$, we define the set of valid cosets
$$
\Xi_{k,|\cS|}^* \triangleq \begin{cases}
    \{C\in\Xi:0\ne C,1\ne C,\max(C)\le n-k-|\cS|\}, &\text{for }k\ge 2,\\
    \{C\in\Xi:1\ne C\}, &\text{for }k = 1.
\end{cases}
$$
We use $C_{i,k,|\cS|}^*$ to denote the $i$-th cyclotomic coset in $\Xi_{k,|\cS|}^*$ and let $s_{i,k,|\cS|}^* = |C_{i,k,|\cS|}^*|$. Hence, the basis of $\cV_{k,\cS}$ is simply the union of the basis polynomials corresponding to the cyclotomic cosets in $\Xi_{k,|\cS|}^*$. That is, 
$$
\cT_{k,\cS}^{*} \triangleq \left\{T_{\ell}^{(i)}:\ell\in[0,s_{i,k,|\cS|}^*-1], i\in[|\Xi_{k,|\cS|}^*|]\right\}.
$$
Thus, we have a similar result as in Theorem~\ref{thm:dim_of_Wk}, but for nonempty $\cS$. We summarize in Theorem~\ref{thm:dim_of_WkS}.
\begin{theorem}\label{thm:dim_of_WkS}
    The dimension of the subspace $\cW_{k,\cS}$ is exactly the sum of the size of each valid cyclotomic coset:
    $$
    \dim(\cW_{k,\cS}) = \dim(\cV_{k,\cS}) = \sum_{i=1}^{|\Xi_{k,|\cS|}^*|} s_{i,k,|\cS|}^*.
    $$
\end{theorem}
For nonempty $\cS$, we can see that the larger the set $\cS$, the stronger the restriction for a cyclotomic coset to be in $\Xi_{k,|\cS|}^*$ which leads to a smaller dimension. As per mentioned in Theorem~\ref{thm:main_theorem}, we can omit $|\cS| + \dim(\cW_{k,\cS})$ number of nodes and we want to maximize this. That is, given $\cS$, let $R_\cS$ be the number of nodes to omit, i.e., $$R_{\cS} \triangleq|\cS| + \dim({\cW_{k,\cS}}),$$
and we want to find $\max_{|\cS|} R_\cS$. We can find $\max_{|\cS|} R_\cS$ by iterating each $|\cS|$. However, we observe that $\dim(\cW_{k,\cS})$ is a non-increasing step function of $|\cS|$. That is, $\dim{(\cW_{k,\cS})}$ decreases only when the constraint $n-k-|\cS|$ drops below the maximum value of all cyclotomic cosets, forcing that coset to be omitted from the basis construction. Hence, instead of iterating all $|\cS|$, we iterate through all valid cyclotomic cosets by pruning. 

In the first iteration, we consider the largest possible set of cyclotomic cosets $\Xi^*_k$ and iteratively remove the coset with the largest value among all cosets. Throughout these iterations, we gradually decrease $\dim(\cW_{k,\cS})$ and increase $|\cS|$, and check which combinations gives us the largest $R_\cS$.

\noindent \textbf{Algorithm O. (Optimization of the Exclusion Set for $k\ge 2$)}
\begin{enumerate}[label=\textbf{O\arabic*.}, leftmargin=*, labelsep=0.5em]
    \item \textbf{(Initialization)} Set $\cU\leftarrow\Xi_k^* = \{C\in\Xi:0\ne C, 1 \ne C, \max(C)\le n-k\}$. Initialize a list to store tuples of $(\cU, |\cS|, R_{\cS})$.
    \item \textbf{(Prune Cosets)} While $\cU$ is nonempty:
    \begin{enumerate}[label = \arabic*.]
        \item Identify the coset $C_{\text{max}}\in\cU$ that contains $\beta = \max_{C\in\cU}\max(C)$.
        \item Update $|\cS| = n-k-\beta$. With this choice of $|\cS|$, the set $\cU$ is precisely $\Xi_{k,|\cS|}^*$ since for any $C\in\cU$, $\max(C)\le \beta = n - k - |\cS|$. 
        \item Compute $R_{\cS} = |\cS| + \sum_{C\in\cU}|C|$.
        \item Record the configuration $(\cU,|\cS|, R_{\cS})$.
        \item Prune $\cU \leftarrow \cU \setminus \{C_\text{max}\}$.
    \end{enumerate}
    \item \textbf{(Boundary Case)} Record $(\cU = \emptyset, |\cS| = n-k-q^{t-1},R_\cS = n-k-q^{t-1})$.
    \item \textbf{(Selection)} Choose the configuration such that $R_\cS$ is maximized.
\end{enumerate}

\begin{remark}
    For $k = 1$, the algorithm proceeds almost identically as the above with a single difference on the first pruning iteration. This is because, we include $\{0\}$ into $\cU \leftarrow \Xi_1^*$. The basis polynomial corresponding to this is $T(x) = 1$ and hence $h(x) = 1/x = x^{n-2}$ satisfies the Degree Condition in Definition~\ref{def:trace-repair-compatible}. Since $\deg(h(x)) = n-2$ and the degree bound is $n-2$, we cannot further multiply $T(x)/x$ with a $g(x)$. This forces $|\cS| = 0$. Consequently, the coset $\{0\}$ imposes the strictest constraint to the space $\cV_{k,\cS}$. Hence, in the first pruning iteration, we prune $\{0\}$ from $\cU$.
\end{remark}

\begin{figure*}[ht!]
    \centering

    \begin{subfigure}[b]{0.32\textwidth}
        \centering
        \includegraphics[width=\linewidth]{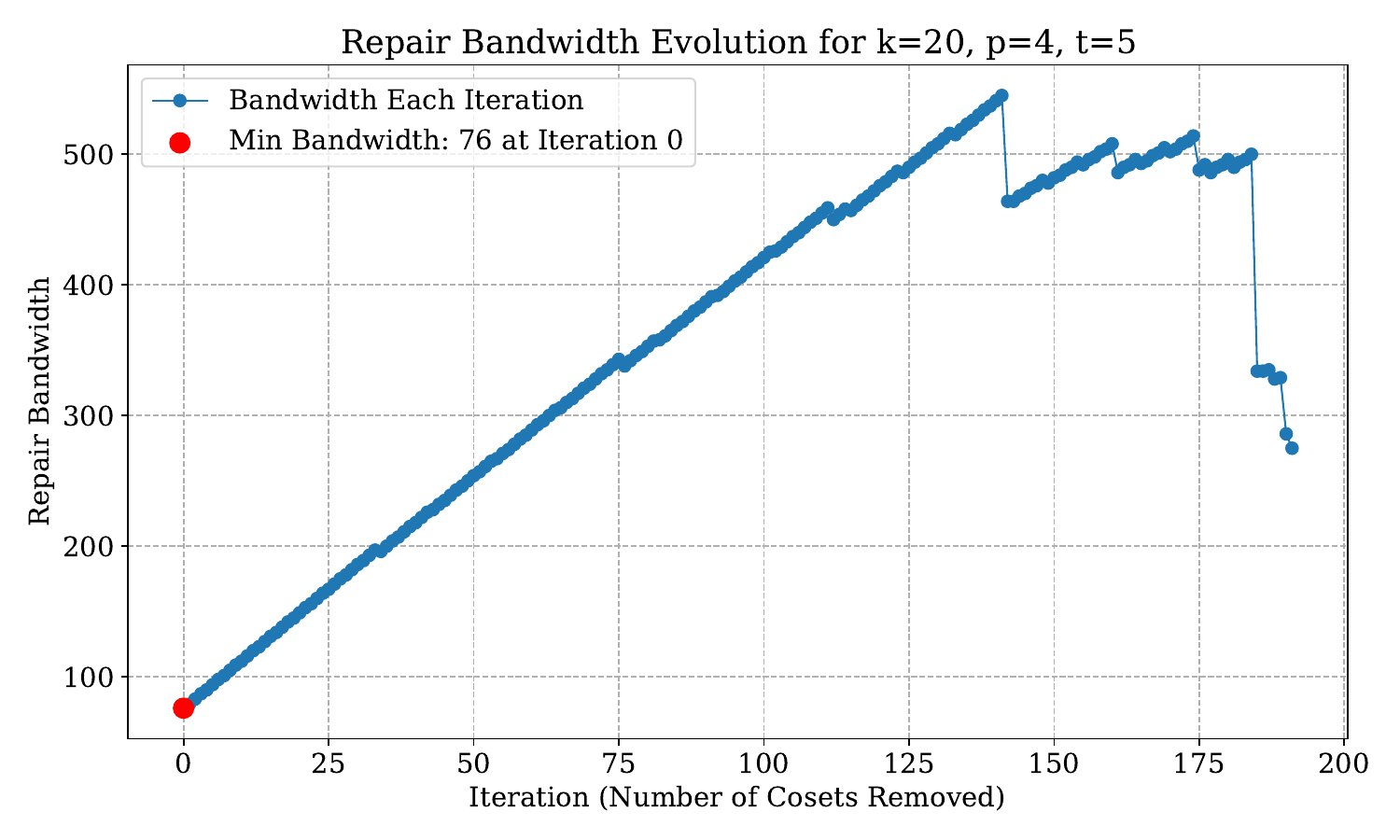}
        \caption{$k=20$.}
        \label{fig:small_k_20}
    \end{subfigure}
    \hfill
    \begin{subfigure}[b]{0.32\textwidth}
        \centering
        \includegraphics[width=\linewidth]{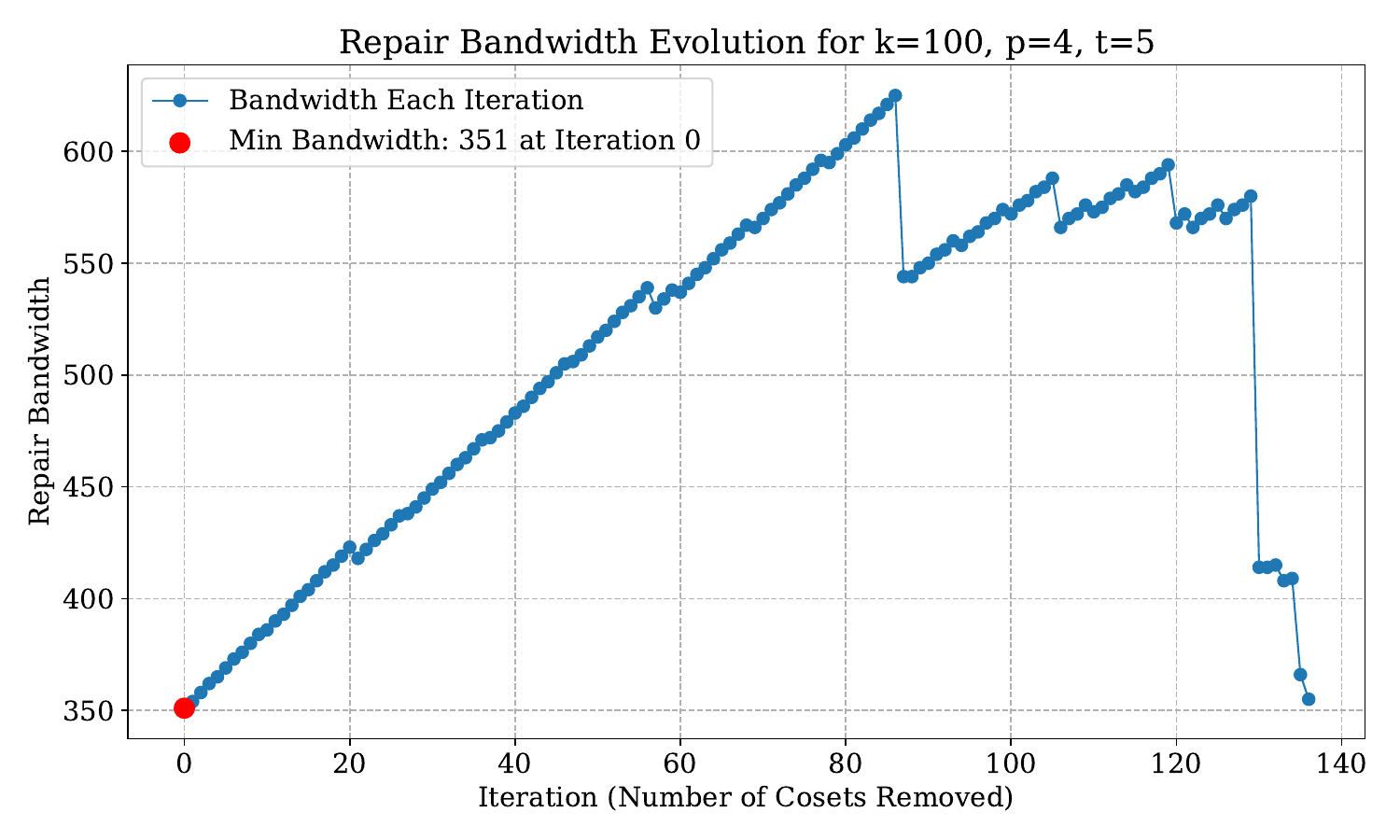}
        \caption{$k=100$.}
        \label{fig:big_k_100}
    \end{subfigure}
    \hfill
    \begin{subfigure}[b]{0.32\textwidth}
        \centering
        \includegraphics[width=\linewidth]{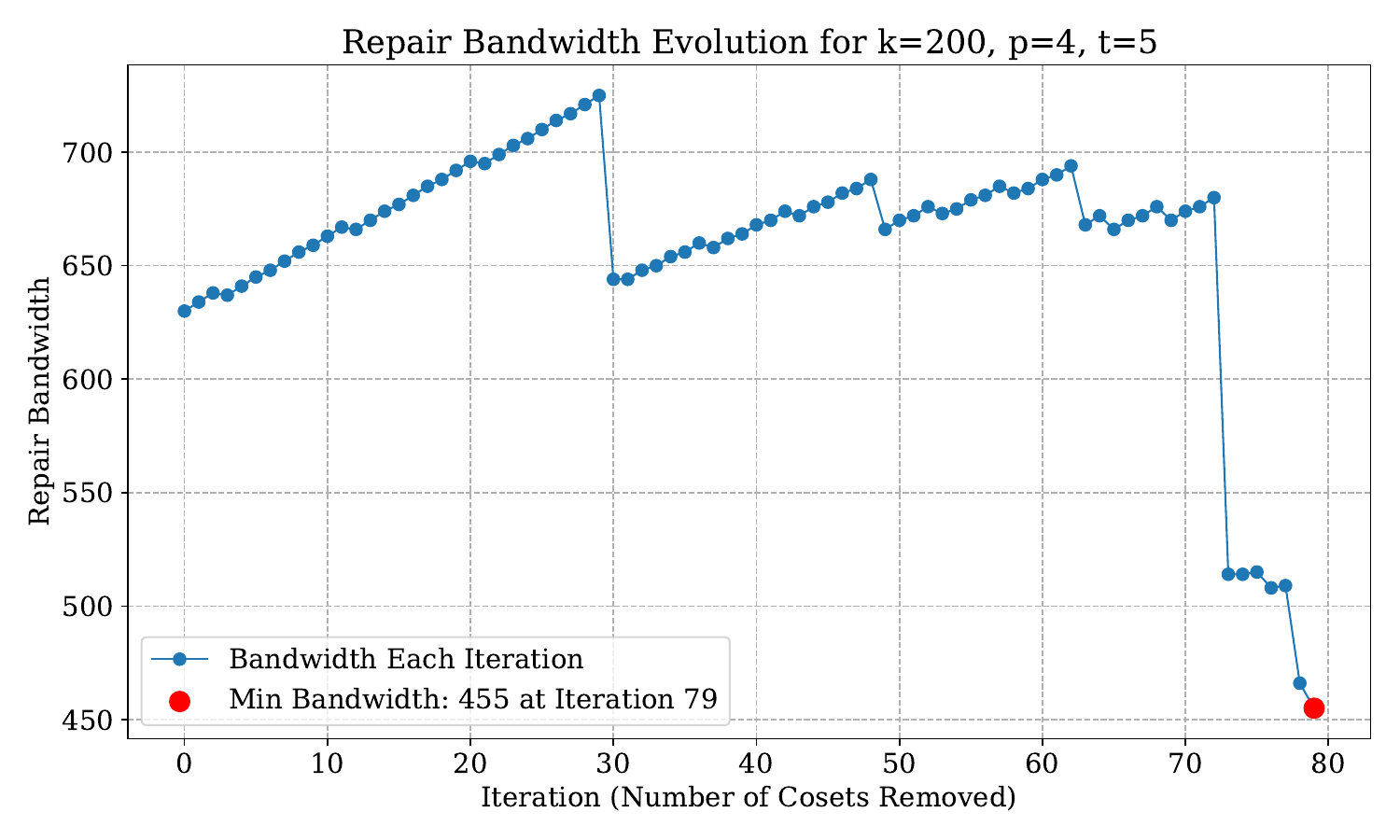}
        \caption{$k=200$}
        \label{fig:big_k_200}
    \end{subfigure}
    
    \caption{Visual comparison of the repair bandwidth as the pruning progresses for various $k$ values with $q = 4$ and $t = 5$. The horizontal axis represents the number of cosets pruned from $\Xi_k^*$. The vertical axis represents the resulting repair bandwidth, while the red dot represents the minimum repair bandwidth. We observe that the graph is shifted as $k$ increases.}
    \label{fig:algo_evo_1}
\end{figure*}

\begin{figure*}[h!]
    \centering

    \begin{subfigure}[b]{0.32\textwidth}
        \centering
        \includegraphics[width=\linewidth]{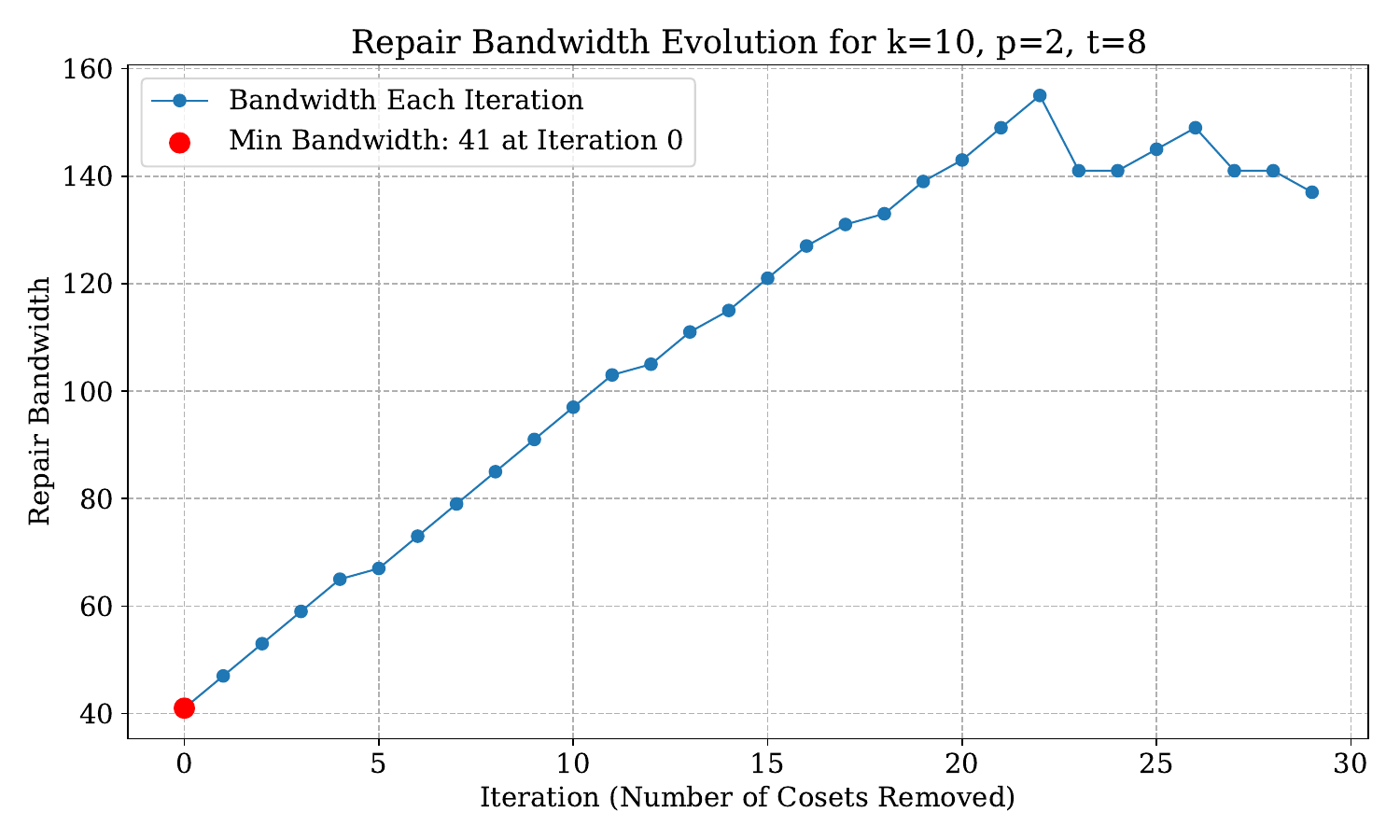}
        \caption{$\BB = \GF(2),\FF =\GF(2^8)$.}
        \label{fig:q=2}
    \end{subfigure}
    \hfill
    \begin{subfigure}[b]{0.32\textwidth}
        \centering
        \includegraphics[width=\linewidth]{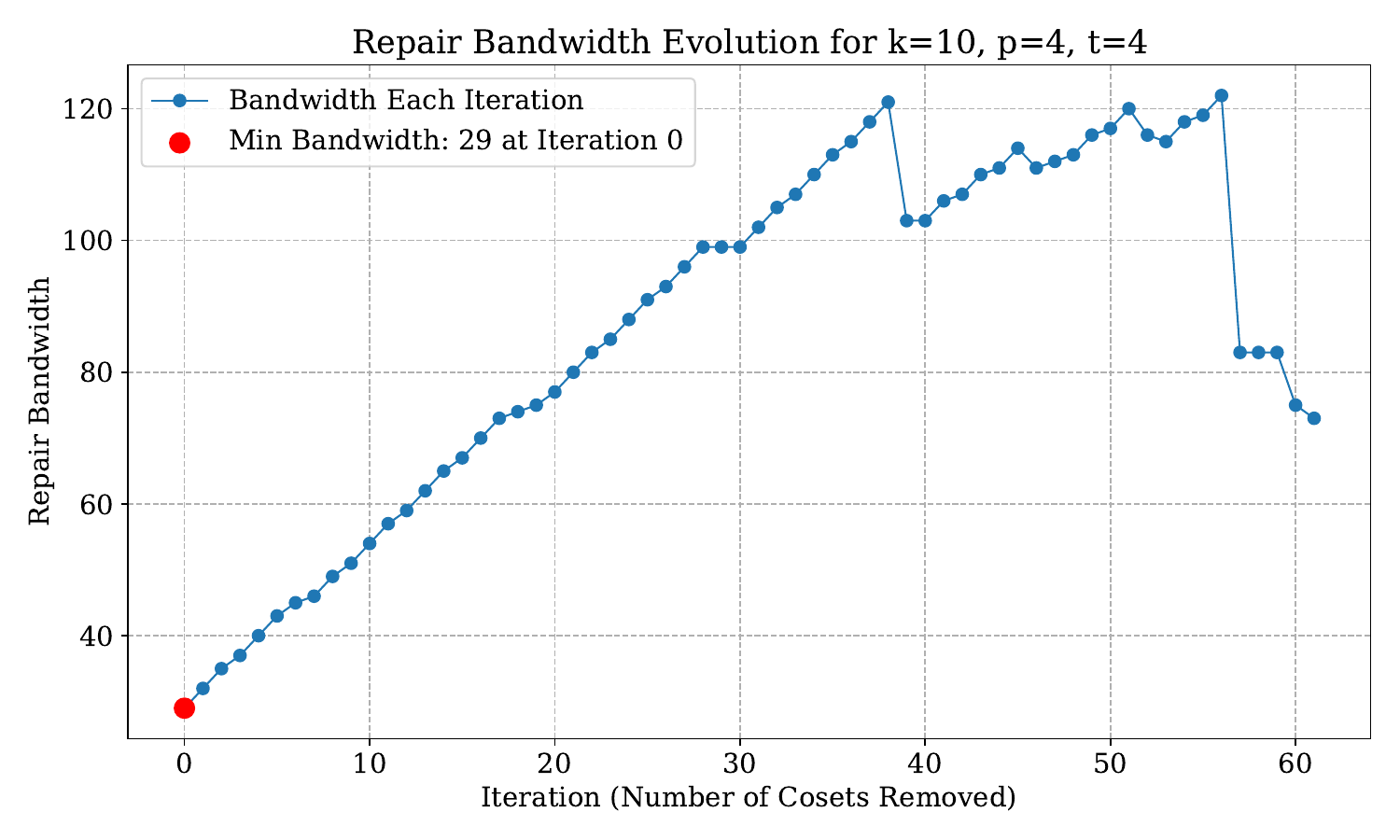}
        \caption{$\BB = \GF(4), \FF = \GF(4^4)$.}
        \label{fig:q=4}
    \end{subfigure}
    \hfill
    \begin{subfigure}[b]{0.32\textwidth}
        \centering
        \includegraphics[width=\linewidth]{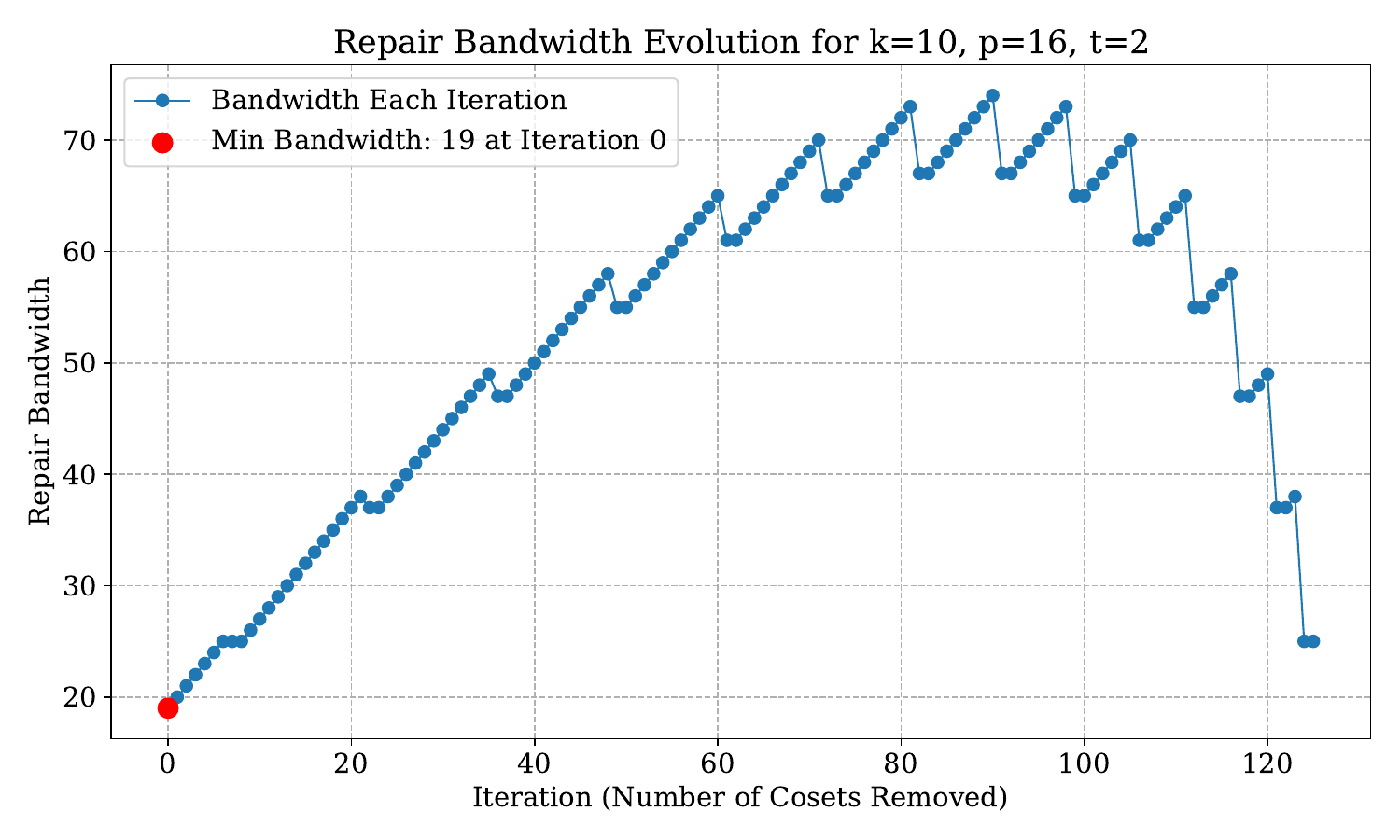}
        \caption{$\BB = \GF(16),\FF = \GF(16^2)$}
        \label{fig:q=16}
    \end{subfigure}
    
    \caption{Visual comparison of the repair bandwidth as the pruning progresses for a fixed $k = 10$ and $\FF = \GF(256)$ but varying $t$. The horizontal axis represents the number of cosets pruned from $\Xi_k^*$. The vertical axis represents the resulting repair bandwidth, while the red dot represents the minimum repair bandwidth. We observe that the graph looks more convex-like as $|\BB|$ increases.}
    \label{fig:algo_evo_2}
\end{figure*}

\section{Explicit Set of Helper Nodes}

\noindent In this section, we show that we can choose which helper nodes to download. We show this formally.  


\begin{corollary}\label{cor:nodes_to_download}
    With the output $(\Xi_{k,|\cS|}^*,|\cS|,R_{\cS})$ from Algorithm~O, consider basis element $\cT_{k,\cS}^*$ for $\cV_{k,\cS}$. Let $d = \dim(\cV_{k,\cS}) = \dim(\cW_{k,\cS})$. Fix $r$ and set $\cI = \{\omega^r,\ldots,\omega^{r+d-1}\}$ and $\cS$ to be any subset of $\FF^*\setminus\cI$ of the given size. Then, by downloading $\tr(g(\alpha)f(\alpha)/\alpha)$ for all $\alpha\in\FF^*\setminus(\cI\cup\cS)$, it is possible to recover $\tr(g(\alpha)f(\alpha)/\alpha)$ for all $\alpha\in\cI$. 
    Hence, we repair $f(0)$ with bandwidth $(n-1-R_\cS)\log|\BB|$ bits.
\end{corollary}
\begin{proof}
    To prove this, we refer back to the proof of Theorem~1. Note that, to perform the repair, we only need to ensure that the matrix $\boldsymbol{M}_\cI$ is invertible with some choice of $\cI$. Indeed, with $\cI = \{\omega^r,\ldots,\omega^{r+d-1}\}$, we have
    \begin{align*}
    \boldsymbol{M}_\cI&= \boldsymbol{VE} = \begin{bmatrix}
        \boldsymbol{V}_1 &\boldsymbol{0} & \cdots & \boldsymbol{0}\\
        \boldsymbol{0} & \boldsymbol{V}_2 & \cdots & \boldsymbol{0}\\
        \vdots & \vdots & \ddots & \vdots\\
        \boldsymbol{0} & \boldsymbol{0} & \cdots & \boldsymbol{V}_{|\Xi_k^*|}
    \end{bmatrix}\boldsymbol{E},
    \end{align*}
    where
    \begin{align*}
    \boldsymbol{E} = \begin{bmatrix}
        \left(\omega^{a^{(1)}}\right)^r & \cdots & \left(\omega^{a^{(1)}}\right)^{r+d-1}\\
        \vdots & \ddots & \vdots \\
        \left(\omega^{a^{(1)}q^{s_1-1}}\right)^r & \cdots & \left(\omega^{a^{(1)}q^{s_1-1}}\right)^{r+d-1}\\
        \vdots & \ddots & \vdots\\
        \left(\omega^{a_{1}^{(|\Xi_k^*|)}}\right)^r & \cdots & \left(\omega^{a_{1}^{(|\Xi_k^*|)}}\right)^{r+d-1}\\
        \vdots & \ddots & \vdots \\
        \left(\omega^{a^{(|\Xi_k^*|)}q^{s_{|\Xi_k^*|}-1}}\right)^r & \cdots & \left(\omega^{a^{(|\Xi_k^*|)}q^{s_{|\Xi_k^*|}-1}}\right)^{r+d-1}
    \end{bmatrix},
    \end{align*}
    and
    \begin{align*}
        \boldsymbol{V}_i = \begin{bmatrix}
            1 & \omega^{a^{(i)}} & \left(\omega^{a^{(i)}}\right)^2 & \cdots & \left(\omega^{a^{(i)}}\right)^{s_i-1}\\
            1 & \omega^{a^{(i)}q} & \left(\omega^{a^{(i)}q}\right)^2 & \cdots & \left(\omega^{a^{(i)}q}\right)^{s_i-1}\\
            \vdots & \vdots & \vdots & \ddots & \vdots\\
            1 & \omega^{a^{(i)}q^{s_i-1}} & \left(\omega^{a^{(i)}q^{s_i-1}}\right)^2 & \cdots & \left(\omega^{a^{(i)}q^{s_i-1}}\right)^{s_i-1}\\
        \end{bmatrix}. 
    \end{align*}
    Clearly, both $\boldsymbol{V}$ and $\boldsymbol{E}$ are invertible. This is because $\boldsymbol{V}$ is a block-diagonal matrix of Vandermonde matrices and $\boldsymbol{E}$ is a Vandermonde matrix. Then, the result follows. 
\end{proof}

With the explicit choice of helper nodes established in Corollary~\ref{cor:nodes_to_download}, we can now summarize the complete repair procedure.

\noindent \textbf{Algorithm X. (Complete Repair Procedure)}. Repair of $c(0)$ with $(n-1-R_\cS)$ symbols of $\BB$.
\begin{enumerate}[label = \textbf{X\arabic*.}, leftmargin=*, labelsep = 0.5em]
    \item \textbf{(Parameter Optimization)} Run \textbf{Algorithm O} to obtain $\Xi_{k,|\cS|}^*$ and $|\cS|$ that maximize $R_\cS$. Let $d = \dim(\cW_{k,\cS}) = \sum_{C\in \cW_{k,\cS}} |C|$.
    \item \textbf{(Excluded Nodes Selection)} Set $\cI = \{1,\omega,\ldots, \omega^{d-1}\}$ and $\cS\subset \FF^*\setminus\cI$ of size $|\cS|$ obtained from \textbf{Algorithm O}.
    \item \textbf{(Trace Retrieval)} Download traces $\tr(g(\alpha)c(\alpha)/\alpha)$ for $\alpha\in\FF^*\setminus(\cI\cup\cS)$.
    \item \textbf{(Trace Reconstruction)} Reconstruct traces $\tr(g(\alpha)c(\alpha)/\alpha)$ for $\alpha\in\cI$ as follows.
    \begin{enumerate}[label = \arabic*.]
        \item Retrieve Trace Repair Compatible polynomials $T_{\ell}^{(i)}$, $\ell\in [0,s_{i,k}^*-1]$ corresponding to the cylcotomic cosets $C_{i,k}^*$ in $\Xi_{k,|\cS|}^*$ for all $i\in[|\Xi_{k,|\cS|}^*|]$.
        \item Set $h_{\ell}^{(i)}(x) = g(x)T_\ell^{(i)}(x)/x$ for all $\ell\in[0,s_{i,k}^*-1],i\in[|\Xi_{k,|\cS|}^*|]$. Here $h_{\ell}^{(i)}(\alpha) = 0$ for all $\alpha\in\cS$ due to $g(x)$.
        \item Consider $d$ parity-check equations and solve for $\tr(g(\alpha)c(\alpha)/\alpha)$ for all $\alpha\in\cI$:
        $$
        \sum_{\alpha \in \FF^*\setminus\cS} h_\ell^{(i)}(\alpha)c(\alpha) = 0 \implies \sum_{\alpha \in \cI} T_\ell^{(i)}(\alpha)\tr(g(\alpha)c(\alpha)/\alpha) = - \sum_{\alpha \in \FF^*\setminus(\cI\cup\cS)} T_\ell^{(i)}(\alpha)\tr(g(\alpha)c(\alpha)/\alpha).
        $$
        Corollary~\ref{cor:nodes_to_download} ensures a unique solution. 
    \end{enumerate}
    \item \textbf{(Symbol Recovery)} Repair $c(0)$ as in Lin \cite{Lin2023} as we have all required traces.
    \begin{enumerate}[label = \arabic*.]
        \item Compute, for all $i\in[t]$,
        $$\tr(u_i g(0) c(0)) = -\sum_{\alpha\in\FF\setminus(\cS\cup\{0\})} \tr(u_i\alpha)\tr\left(\frac{g(\alpha) c(\alpha)}{\alpha}\right).$$
        \item Repair
        $$
        c(0) = (g(0))^{-1}\sum_{i\in[t]}u_i^\perp \tr(u_i g(0) c(0)).
        $$
    \end{enumerate}
\end{enumerate}

\section{Performance Analysis and Comparisons}


In this section, we analyze the performance of our optimized trace-repair scheme ({Algorithm X}) and compare it with the classical repair scheme and recent works by Lin \cite{Lin2023} and Liu et. al. \cite{Liu2024}. Figure~\ref{fig:algorithm_comparison_vary_k} illustrates the repair bandwidth (in base field symbols) for varying dimensions $k$ over different fields $\FF$. The curve labeled ``Our Work (Optimized)'' represents the bandwidth achieved by {Algorithm X}.

\subsection{Universal Improvement over Classical Repair}

A central contribution of this work is the guarantee that trace repair is universally efficient for all $k \le n - q^{t-1}$ in bandwidth.

\begin{theorem}\label{thm:outperform}
    Suppose $\cA = \FF$ and fix $k\le n - q^{t-1}$. Then, we can always repair $f(0)$ with at most $k\log |\FF|$ bits.
\end{theorem}

\begin{proof}
        It is sufficient to prove that the framework with $g(x) = 1$ can repair $f(0)$ with at most $k\log|\FF|$ bits. The number of nodes involved in the repair scheme is the total of the number of elements in all cyclotomic coset we remove. Formally, given $k$, the number of nodes involved is
        $$
        n-1-\dim(\cW_k) = \sum_{C\in\Xi}|C| - \sum_{C\in\Xi_{k}^*}|C| = \sum_{C\in\Xi\setminus\Xi_k^*}|C|
        $$
        Note that, each cyclotomic has size at most $t$. Hence,
        $$
        \sum_{C\in\Xi\setminus\Xi_k^*}|C| \le t |\Xi\setminus\Xi_k^*|.
        $$
        \begin{itemize}[leftmargin=*]
        \item When $k\ge 2$, we exclude the cyclotomic coset $\{0\}$, cyclotomic coset with entry $1$, and all cyclotomic cosets with some entry more than $q^t-k$. Since the maximum entry of the coset is $q^t-2$, we remove at most $k$ cyclotomic cosets. In other words, $|\Xi\setminus\Xi_k^*|\le k$ when $k\ge 2$.
        \item When $k=1$, we exclude only one cyclotomic coset with entry $1$.
        So, we also have $|\Xi\setminus\Xi_1^*| = 1 \le k$.
        \end{itemize}
        Hence, the bandwidth of the trace-mapping framework is
        $$
        \left(n-\dim(\cW_k) - 1\right)\log|\BB|\le kt \log|\BB| = k\log|\FF|.
        $$
\end{proof}
As stated in Theorem~\ref{thm:outperform}, our scheme requires a bandwidth of at most $k \log |\FF|$ bits. This ensures that the trace repair framework never consumes more bandwidth than the classical repair scheme, regardless of the dimension $k$. Empirically, this is evident in Figure~\ref{fig:bandwidth_comparison_4^2}, \ref{fig:bandwidth_comparison_8^2}, and \ref{fig:bandwidth_comparison_2^8}, the repair bandwidth of our scheme remains strictly below (when $k>1$) or equal (when $k=1$) to the classical repair bandwidth. While the Guruswami-Wootters scheme \cite{guruswamiwooters2017} only outperforms classical repair when $k > (n-1)/t$, our approach outperforms for all $k>1$.

\subsection{Advancement over State-of-the-Art}

\noindent \underline{\textit{Comparison with Liu et al. \cite{Liu2024}.}} Liu et al. identified that linear dependencies among downloaded traces could reduce bandwidth, relating this reduction to the dimension of a subspace $\cW_k$. However, their work primarily established existence results and provided lower bounds on $\dim(\mathcal{W}_k)$ for specific fields. In contrast, our work exactly quantifies this the dimension. By establishing a linear isomorphism between the subspace $\cW_{k,\cS}$ and the space of Trace-Repair Compatible polynomials $\cV_{k,S}$ (Theorem~\ref{thm:dim_V}), we determine the exact dimension of the subspace $\cW_{k,\cS}$ by constructing a basis for $\cV_{k,\cS}$. This observation also allows us to move beyond existence proofs to an explicit construction of helper nodes (Corollary~\ref{cor:nodes_to_download}), ensuring the promised bandwidth is achievable in practice (see {Algorithm X}). In Fig.~\ref{fig:bandwidth_comparison_4^2} and~\ref{fig:bandwidth_comparison_8^2}, the curve ``Our Work (Case $g(x)=1$)'' corresponds to this exact dimension computation when no zero-forcing is applied ($\cS=\emptyset$), which matches or improves upon the bounds provided by Liu et al.

\noindent\underline{\textit{Comparison with Lin \cite{Lin2023}}.} Lin proposed the zero-forcing method, utilizing a polynomial $g(x)$ to exclude a set of nodes $\cS$ from the repair process entirely. However, Lin's approach is constrained to a fixed set size of $|\cS| = n - k - q^{t-1}$. This rigid choice imposes strict degree constraints that eliminate the possibility of exploiting linear dependencies in the traces (i.e., $\dim(\mathcal{W}_{k,S})$ becomes $0$). Our work generalizes this by treating $|\cS|$ as a tunable parameter. 
Using {Algorithm~O}, we identify a trade-off between the exclusion set size $|\cS|$ and the subspace dimension $\dim(\mathcal{W}_{k,S})$. By iteratively pruning valid cyclotomic cosets, we search for the optimal configuration that maximizes the total number of excluded nodes $R_\cS = |\cS| + \dim(\cW_{k,\cS})$. This optimization is clearly visualized in Figure~\ref{fig:algorithm_comparison_vary_k}. 
Lin's scheme performs well in specific ranges but is outperformed by our optimized scheme in others.

\bibliographystyle{plain}
\bibliography{ref_list.bib}
\end{document}